\documentclass[prx,aps,amssymb,twocolumn,superscriptaddress, notitlepage, longbibliography]{revtex4-2}
\usepackage{bm,amssymb,dsfont,mathrsfs,amstext,latexsym,physics,relsize}
\usepackage{amsthm}
\usepackage{amsmath}
\usepackage{graphicx}
\usepackage[usenames,dvipsnames]{xcolor}
\usepackage[colorlinks=true,citecolor=blue,urlcolor=blue,linkcolor=blue]{hyperref}
\usepackage{orcidlink}
\usepackage{braket}
\usepackage{txfonts} 
\usepackage[usenames,dvipsnames]{xcolor}
\usepackage{enumitem}
\usepackage{multirow}

\newtheorem{lem}{Lemma}
\renewcommand{\textcolor}[2]{#2}
\begin{document}
\title{The Time-Dependent Jastrow
Ansatz: Exact Quantum Dynamics, Shortcuts to Adiabaticity, and Quantum Quenches
in Strongly-Correlated Many-Body Systems}
\author{Jing Yang\orcidlink{0000-0002-3588-0832}} 
\email{jing.yang@su.se}
\altaffiliation{Present Address: Nordita, KTH Royal Institute of Technology and Stockholm University, Hannes Alfv\'ens v\"ag 12, 106 91 Stockholm, Sweden.}
\address{Department of Physics and Materials Science, University of Luxembourg,
L-1511 Luxembourg, Luxembourg}

\author{Adolfo~del Campo\orcidlink{0000-0003-2219-2851}}  
\email{adolfo.delcampo@uni.lu}
\affiliation{Department  of  Physics  and  Materials  Science,  University  of  Luxembourg,  L-1511  Luxembourg,  Luxembourg}
\affiliation{Donostia International Physics Center,  E-20018 San Sebasti\'an, Spain}

\begin{abstract}
\textcolor{blue}{
The description of strongly correlated quantum many-body systems far from equilibrium presents a fundamental challenge due to the vast amount of information it requires. We introduce a generalization of the Jastrow ansatz for time-dependent wavefunctions that offers an efficient and exact description of the time evolution of various strongly correlated systems. Previously known exact solutions are characterized by scale invariance, enforcing self-similar evolution of local correlations, such as the spatial density. However, we demonstrate that a complex-valued time-dependent Jastrow ansatz (TDJA) is not restricted to scale invariance and can describe a broader class of dynamical processes lacking this symmetry. The associated time evolution is equivalent to the implementation of a shortcut to adiabaticity (STA) via counterdiabatic driving along a continuous manifold of quantum states described by a real-valued TDJA, providing a framework for engineering exact STA in strongly correlated many-body quantum systems. We illustrate our findings in systems with inverse-square interactions, such as the Calogero-Sutherland and hyperbolic models, supplemented with pairwise logarithmic interactions, as well as in the long-range Lieb-Liniger model, where bosons experience both contact and Coulomb interactions in one dimension. Our results enable the study of quench dynamics in all these models and serve as a benchmark for numerical and quantum simulations of nonequilibrium strongly correlated systems with continuous variables.}
\end{abstract}

\maketitle

\section{Introduction}

\textcolor{blue}{The nonequilibrium dynamics of isolated quantum many-body systems give rise to rich and intriguing physics, encompassing phenomena such as quantum chaos, many-body localization, and thermalization \citep{DAlessio16, abanin2019colloquium}. Understanding these phenomena becomes particularly challenging in the presence of strong interactions, where perturbative approaches fail. In such cases, numerical and analytical techniques are essential, with notable examples including exact diagonalization, the density-matrix renormalization group \citep{schollwock2011thedensitymatrix}, and quantum Monte Carlo algorithms.}

\textcolor{blue}{In one spatial dimension \citep{Cazalilla11, GuanBatchelor13}, additional methods are available for studying strongly correlated systems. For instance, the Bethe ansatz and the quantum inverse scattering method allow for the exact solution of certain integrable many-body quantum systems \citep{KBI97, Gaudin14, Sutherland04}. However, even when the exact eigenstates are known, evaluating correlation functions remains a significant challenge, both in thermal equilibrium and, even more so, in nonequilibrium settings \citep{Rylands20, GuanHe22}.
Furthermore, numerical techniques such as the density-matrix renormalization group and quantum Monte Carlo algorithms face accuracy limitations when simulating long-time unitary dynamics. These difficulties become particularly pronounced in systems with continuous variables.}

\textcolor{blue}{Recent advancements in quantum simulation have made it possible to experimentally realize strongly interacting models, fulfilling Feynman’s vision of quantum emulation \citep{Feynman82}. A paradigmatic platform for analog quantum simulation is provided by ultracold gases \citep{bloch2008manybody}. For example, Olshanii demonstrated that ultracold bosons interacting via $s$-wave scattering are effectively described by the celebrated Lieb-Liniger (LL) model \citep{LL63, L63} when confined in a tight waveguide \citep{Olshanii98}.}

\textcolor{blue}{
Moreover, the implementation of these models in ultracold atomic systems enables precise control over interactions, allowing them to be tuned from the strongly attractive to the strongly repulsive regimes via Feshbach resonance. This tunability has made it possible to explore the quench dynamics of many-body systems in the laboratory, such as those induced by an interaction quench or a change in confinement \citep{Cazalilla11, Langen15, Campbell15, GuanHe22}—a regime where strong correlations pose significant theoretical challenges \citep{Girardeau03, Buljan08, Jukic08, Iyer12, Zill15, Carleo17, Ruggiero20}.}

\textcolor{blue}{
Further progress in quantum simulation has been driven by the digital approach, where the dynamics of interest are approximated using a quantum circuit. This paradigm is under extensive investigation, with applications spanning condensed matter physics \citep{Raeisi12, Smith19}, quantum field theory \citep{Jordan12}, and quantum chemistry \citep{Cao19}. Additionally, hybrid analog-digital approaches offer an alternative framework for exploring nonequilibrium quantum phenomena \citep{Lamata18}.
However, progress in quantum simulation is constrained by the scarcity of analytical and exact results for the nonequilibrium dynamics of strongly correlated many-body systems \citep{Mitra18, GuanHe22}. These results are not only valuable from a fundamental perspective but also serve as crucial benchmarks for quantum simulation algorithms and provide insights into complex experimental observations.}

\textcolor{blue}{
Many interacting many-body systems of interest exhibit strong correlations in their ground state \citep{wen2004quantum, mahan2000manyparticle, girvin2019moderncondensed}. One of the simplest wavefunctions that captures these ground-state correlations is the (Bijl-Dingle-) Jastrow ansatz \citep{Bijl40, Dingle49, Jastrow55}, which assumes that the ground-state wavefunction can be expressed as a product of pair functions that depend only on the interparticle spacing. While generally considered an approximate wavefunction for many-body quantum systems, the Jastrow ansatz is particularly useful in describing perturbative expansions at low densities \citep{gaudin1971jastrow} and is widely employed as a trial wavefunction in quantum Monte Carlo techniques.}

\textcolor{blue}{Moreover, the Jastrow ansatz provides an exact description of the ground state for certain interacting many-body systems. To identify such models, one can assume a Jastrow-form wavefunction and determine the corresponding parent Hamiltonian by solving the time-independent Schr\"odinger equation. This approach was pioneered by Calogero \citep{Calogero71} and Sutherland \citep{Sutherland71, Sutherland04}, leading to the discovery of the well-known family of one-dimensional integrable Calogero-Sutherland (CS) models. These models include particles interacting via inverse-square potentials in unbounded space, as well as inverse-square sinusoidal interactions under periodic boundary conditions \citep{Sutherland04, kuramoto2009dynamics}.}

\textcolor{blue}{
Similarly, the attractive Lieb-Liniger (LL) gas has long been known to exhibit bright quantum soliton states with a Jastrow-form wavefunction \citep{McGuire64, BeauPittman20}. By now, the complete family of models describing interacting identical particles with a Jastrow ground state has been established in any spatial dimension \citep{CalogeroMarchioro75, delcampo20, BeauDC21, yang2022one}. These parent Hamiltonians generally include both two-body and three-body interactions. While the latter are absent in the celebrated CS and LL models, renormalization group calculations have shown that three-body interactions are irrelevant for long-wavelength, low-temperature physics \citep{Kane91}, further reinforcing the practical utility of the Jastrow ansatz in physical applications.}

\textcolor{blue}{In contrast to its widespread use in equilibrium settings, the application of the Jastrow ansatz to nonequilibrium scenarios remains relatively unexplored. Notably, the dynamics of the rational Calogero-Sutherland (CS) model—describing one-dimensional bosons with inverse-square interactions confined in a harmonic trap—are exactly captured by a time-dependent Jastrow ansatz at all times \citep{Sutherland98, DaeYupSong01, DaeYupSong02}. This model not only serves as a valuable testbed for studying nonequilibrium phenomena but also includes hard-core bosons in the Tonks-Girardeau regime as a limiting case \citep{Girardeau60, Girardeau01, minguzzi2005exactcoherent, GangardtPustilnik08}, making it highly relevant to ultracold atom experiments \citep{Kinoshita04, Paredes04, Wilson20}. 
As a result, this model has inspired a wide range of studies on topics such as nonexponential decay, quantum speed limits, Loschmidt echoes, and orthogonality catastrophe \citep{delcampo16, Fogarty20, XuLi20, LvZhang20, delcampo21}. It has also played a key role in advancing research in finite-time quantum thermodynamics \citep{Jaramillo16, Beau16, GarciaMarch16, BeauDelcampo20, Atas20}, quantum quenches \citep{Rajabpour14, Campbell15}, and quantum control \citep{delcampo11, delCampo12, Dupays21, Dupays22}, among other areas.}

\textcolor{blue}{
While these applications help to illustrate the value of exact solutions
in quantum dynamics, they are all characterized by scale-invariance.
This dynamical symmetry is highly restrictive and results in self-similar 
time evolution, meaning that the spatial probability densities (the absolute 
square of the wavefunction in the coordinate representation) at any two 
different times are simply related by a rescaling of the coordinate variables. 
For instance, in one spatial dimension, the time dependence of a many-body 
quantum state $\Psi(t)$ satisfies  
\begin{equation}
|\Psi(x_{1},\dots,x_{N},t)|^{2}=\frac{|\Psi(x_{1}/b(t),\dots,x_{N}/b(t),t=0)|^{2}}{b(t)^{N}},\label{densityinv}
\end{equation}
where $b(t)>0$ is the scaling factor. In this case, the complexity 
of many-body time evolution reduces significantly to determining the 
scaling factor, which obeys an ordinary differential equation known 
as the Ermakov equation. The latter was first introduced in the context 
of the time-dependent harmonic oscillator~\citep{LewisRiesenfeld69,Chen10}, 
and together with its generalizations, it plays a fundamental role 
in the study and control of Bose-Einstein condensates~\citep{CastinDum96,Kagan96,gritsev2010scaling} 
and ultracold Fermi gases~\citep{Castin04,WernerCastin06,Deng16,Deng18,Deng18Sci,Diao18}.}  

\textcolor{blue}{
Beyond the realm of scale-invariant dynamics, results are scarce, and 
time-dependent Jastrow ans\"atze have only recently been explored in numerical 
methods. Notably, integrating the Jastrow ansatz with quantum Monte Carlo 
algorithms has been applied to the study of the quench dynamics of the 
Lieb-Liniger (LL) model~\citep{Carleo17}. } 

\textcolor{blue}{A natural question arises: can the construction of the parent Hamiltonian 
used for the stationary Jastrow ansatz be extended to the time-dependent 
setting for arbitrary processes? However, a direct extension of this 
approach generally leads to a parent Hamiltonian that is not necessarily 
Hermitian. The underlying reason for this non-Hermiticity is that the 
dynamics implicitly assumed by the time-dependent trial wavefunction may 
break unitarity. It is worth noting that recent progress in finding 
parent Hamiltonians~\citep{Rattacaso21} for time-dependent quantum 
states has primarily focused on discrete spin systems, which are not 
directly applicable to the continuous-variable many-body systems considered here. } 

\textcolor{blue}{In this work, we extend the program of constructing parent Hamiltonians 
for Jastrow wavefunctions to the time-dependent case in one-dimensional 
quantum many-body systems in the continuum, i.e., with continuous variables. 
We derive consistency conditions for the one-body and two-body pair functions 
that define the Jastrow ansatz and apply them to systems embedded in a 
harmonic trap. These consistency conditions lead to the parent Hamiltonian 
of the complex-valued time-dependent Jastrow ansatz (TDJA), where the 
amplitude of the two-body trial wavefunction retains a functional form 
similar to that of the ground state of the Calogero-Sutherland (CS), 
Hyperbolic, and attractive Lieb-Liniger models.}  

\textcolor{blue}{Interestingly, while the Ermakov equation—ubiquitous in the study of 
scale-invariant dynamics with specific interactions~\citep{Castin04,gritsev2010scaling,delcampo16}—
emerges in our approach, the dynamics described by the complex-valued TDJA 
are not necessarily scale-invariant. Applying our framework to the aforementioned 
trial wavefunctions leads to generalizations of the CS, Hyperbolic, and 
LL models with additional interactions. 
In these systems, we establish the relations between the initial 
density distribution and its time evolution, as well as the long-time momentum 
distribution.}

\textcolor{blue}{As a further application of our results, we demonstrate that the parent Hamiltonian 
of the complex-valued TDJA can be utilized for the engineering of Shortcuts 
to Adiabaticity (STA)~\citep{Torrontegui13,guery-odelin2019shortcuts}, which 
enable the fast preparation of a target state from a given initial eigenstate 
without requiring the long timescales necessary for traditional adiabatic evolution.  
More specifically, we show that the parent Hamiltonian of the complex-valued TDJA 
can be interpreted as an implementation of the adiabatic evolution of the real-valued  TDJA. Furthermore, we illustrate how our findings can be applied to study the 
exact dynamics following a quantum quench of the interparticle interactions.  
In particular, we demonstrate this in the context of the long-range 
Calogero-Sutherland models with logarithmic interactions and the long-range 
Lieb-Liniger model with Coulomb interactions.}

\section{The Time-dependent Jastrow ansatz (TDJA)}

The original Jastrow ansatz~\citep{Bijl40,Dingle49,Jastrow55}, constructed
in terms of products of a pair function and a one-body function, is
time-independent and real-valued. Its use turned out to be a fruitful
approach and led to the discovery of many integrable models, including
the family of Calogero-Sutherland models \citep{Calogero71,Sutherland71,CalogeroMarchioro75,delcampo20,BeauDC21}.
More recently, other models with ground state wave function supporting
a real-valued time-independent Jastrow ansatz (TIJA) have been discovered~\citep{yang2022one,delcampo20,BeauPittman20}.

In this section, we focus on the generalization describing the time-dependent
case and introduce the TDJA, 
\begin{equation}
\Psi(\bm{x},t)=\frac{1}{\exp[\mathcal{N}(t)]}\prod_{i<j}f_{ij}(t)\prod_{k}g_{k}(t),\label{eq:TD-Jastrow}
\end{equation}
where $\bm{x}=(x_{1},\,x_{2},\,\cdots,\,x_{N})$, $\exp[\mathcal{N}(t)]$
is the normalization of the Jastrow wave function, where $\mathcal{N}(t)$
is a real-valued function that only depends on time. 
As a shorthand, we  also define $d^Nx=\prod_{m=1}^{N}dx_{m}$ for short. In addition, $f_{ij}(t)\equiv f(x_{ij},\,t)$
and $g_{k}(t)\equiv g(x_{k},\,t)$ are functions of the particles'
coordinates and time. Throughout this work, we focus on bosonic wave
functions exclusively so that $f(x,\,t)=f(-x,\,t)$. The TDJA~(\ref{eq:TD-Jastrow})
describes the exact solution to the time-dependent Schr\"odinger equation
\begin{equation}
\text{i}\hbar\dot{\Psi}(\bm{x},\,t)=\hat{\mathscr{H}}(t)\Psi(\bm{x},\,t),\label{TISEeq}
\end{equation}
when the dynamics is generated by the many-body Hamiltonian 
\begin{eqnarray}
\hat{\mathscr{H}}(t)&=&\sum_{i}\frac{\hat{p}_{i}^{2}}{2m}+\frac{\hbar^{2}}{2m}\sum_{i}v_{1}^{(i)}+
\frac{\hbar^{2}}{m}
\sum_{i<j}v_{2}^{(ij)}\nonumber\\
& & +
\frac{\hbar^{2}}{m}
\sum_{\substack{i<j<k}
}v_{3}^{(ijk)}-\text{i}\hbar\dot{\mathcal{N}}(t),\label{eq:scrH-def}
\end{eqnarray}
where 
\begin{align}
v_{1}^{(i)} & \equiv\frac{g_{i}^{\prime\prime}}{g_{i}}+\frac{2\text{i}m}{\hbar}\frac{\dot{g}_{i}}{g_{i}},\label{eq:v1-def}\\
v_{2}^{(ij)} & \equiv\frac{f_{ij}''}{f_{ij}}+\left(\frac{g_{i}^{\prime}}{g_{i}}-\frac{g_{j}^{\prime}}{g_{j}}\right)\frac{f_{ij}^{\prime}}{f_{ij}}+\frac{\text{i}m}{\hbar}\frac{\dot{f}_{ij}}{f_{ij}},\label{eq:v2-def}\\
v_{3}^{(ijk)} & \equiv-\left(\frac{f_{ij}'f_{jk}^{\prime}}{f_{ij}f_{jk}}+\frac{f_{ij}'f_{ki}^{\prime}}{f_{ij}f_{ki}}+\frac{f_{ki}'f_{jk}^{\prime}}{f_{ki}f_{jk}}\right),\label{eq:v3-def}
\end{align}
are the normalized one-body, two-body, and three-body potentials bearing
the dimension of inverse length squared. Throughout the work, the
overdot denotes the time derivative, while the derivative with respect to a spatial coordinate of a function $f$ is denoted by $f'$.

We note that the many-body time-dependent potential in Eq.~(\ref{eq:scrH-def})
is not Hermitian in general. Once the Hermiticity is guaranteed, the
unitarity of the dynamics dictates that the norm of $\Psi$ must be
constant in time. {This, in turn, determines the time-dependence of a
the  multi-dimensional integral over the spatial coordinates up
to a constant, i.e.,} 
\begin{equation}
\exp[2\mathcal{N}(t)]\propto\int d^Nx\prod_{i<j}|f_{ij}(t)|^{2}\prod_{k}|g_{k}(t)|^{2}.\label{eq:normalization}
\end{equation}
However, we note that imposing a constant norm of the trial wave function
alone is not sufficient to yield a Hermitian Hamiltonian. 
Additional constraints are required on
on the functional forms of $f_{ij}(t)$ and $g_{k}(t)$, as discussed in Sec.~\ref{sec:Imposing-Hermicitiy}.

The family of Hamiltonians (\ref{eq:scrH-def}) provides a generalization
to driven systems of the seminal result known in the stationary case,
i.e., the family of parent Hamiltonians with stationary real-valued
Jastrow ground state \citep{CalogeroMarchioro75,Sutherland04,delcampo20,BeauDC21}.
Naturally, this family and the corresponding real-valued TIJA as the
ground state can be recovered~ by choosing $f_{ij}$ and $g_{k}$
time-independent and real-valued in Eqs. (\ref{eq:scrH-def})~and
(\ref{eq:TD-Jastrow}), respectively. 

We consider two 
options in the time-dependent case.
The most general option is to make $f_{ij}(t)$ and $g_{k}(t)$ both
time-dependent and complex-valued. Since any complex number can be
represented in the polar form by a non-negative amplitude multiplied
by a phase, without loss of generality, we take 
\begin{equation}
f_{ij}(t)=e^{\Gamma_{ij}(t)+\text{i}\theta_{ij}(t)},\:g_{k}(t)=e^{\Lambda_{k}(t)+\text{i}\phi_{k}(t)},\label{eq:complex-trial-WF}
\end{equation}
where we use the compact notation $\Gamma_{ij}(t)\equiv\Gamma(x_{ij},\,t)$,
$\Lambda_{k}(t)\equiv\Lambda(x_{k},t)$, etc. Here, $\theta_{ij}(t)$
is a two-body phase, and $\phi_{k}(t)$ is a one-body
phase. Note that the zero-body phase $\tau(t)$ can be
incorporated into the phase factor $e^{\text{i}\phi_{k}(t)}$ by redefining
$\phi_{k}(t)\to\phi_{k}(t)=\phi_{k}(t)+\tau(t)/N$ and hence is omitted.

We propose to reverse engineer the parent Hamiltonian using the time-dependent
 Schr\"odinger  equation (\ref{TISEeq}) with the Jastrow ansatz as an
exact solution. Specifically, we shall solve for $\mathscr{H}(t)$
and demonstrate that the solution admits a variety of applications, including
the use of STA and quench dynamics for several one-dimensional many-body
strongly correlated quantum models.

Before doing so, we note that an alternative ansatz can be constructed
by considering $f_{ij}$ and $g_{k}$ to be time-dependent, while
keeping them real-valued, i.e., 
\begin{equation}
\Phi(\bm{x},t)=\frac{1}{\exp[\mathcal{N}(t)]}\prod_{i<j}e^{\Gamma_{ij}(t)}\prod_{k}e^{\Lambda_{k}(t)},\label{eq:TDIR-Jastrow}
\end{equation}
which we shall call real-valued TDJA. The parent Hamiltonian of the
real-valued TDJA according to the time-dependent Schr\"odinger equation
\begin{equation}
\text{i}\hbar\dot{\Phi}(\bm{x},\,t)=\hat{\mathscr{H}}^{\prime}(t)\Phi(\bm{x},\,t)
\end{equation}
can also be found analogously. Clearly, 
\begin{equation}
\Psi(\bm{x},t)=\mathscr{U}_{P}(\bm{x},t)\Phi(\bm{x},t),
\end{equation}
where the many-particle phase unitary operator is defined as 
\begin{equation}
\mathscr{U}_{P}(\bm{x},t)\equiv\prod_{i<j}e^{\text{i}\theta_{ij}(t)}\prod_{k}e^{\text{i}\phi_{k}(t)}.
\end{equation}
The corresponding Hamiltonians $\hat{\mathscr{H}}^{\prime}(t)$ and
$\mathscr{\hat{H}}(t)$ are unitarily equivalent in the sense that
\begin{equation}
\mathscr{\hat{H}}^{\prime}(t)=\mathscr{U}_{P}^{\dagger}(\bm{x},t)\mathscr{\hat{H}}(t)\mathscr{U}_{P}(\bm{x},t)-\text{i}\hbar\mathscr{U}_{P}^{\dagger}(\bm{x},t)\dot{\mathscr{U}_{P}}(\bm{x},t).\label{eq:U-connection}
\end{equation}

In what follows, we  find $\hat{\mathscr{H}}(t)$ first and subsequently determine
$\hat{\mathscr{H}}^{\prime}(t)$ through Eq.~(\ref{eq:U-connection}).
 We will see that the many-body potential in $\hat{\mathscr{H}}^{\prime}(t)$
is non-local, i.e., it involves a term linear in the particles' momenta.
Nevertheless, one can still define the prime parent
Hamiltonian according to the instantaneous time-independent Schr\"odinger
equation 
\begin{equation}
\hat{\mathscr{H}}_{0}^{\prime}(t)\Phi(\bm{x},t)=0,
\end{equation}
i.e., $\Phi(\bm{x},t)$ is the instantaneous eigenstate of $\hat{\mathscr{H}}_{0}^{\prime}(t)$
with zero eigenvalue. Note that if $\Phi(\bm{x},t)$ has no nodes, then it is
generally the ground state of $\hat{\mathscr{H}}_{0}^{\prime}(t)$
provided $\mathscr{\hat{H}}_{0}^{\prime}(t)$ is bounded from below.
The procedure of finding $\hat{\mathscr{H}}_{0}^{\prime}(t)$ leads
to 
\begin{align}
\hat{\mathscr{H}}_{0}^{\prime}(t) & =\sum_{i}\frac{\hat{p}_{i}^{2}}{2m}+\frac{\hbar^{2}}{2m}\sum_{i}(\Lambda_{i}^{\prime\prime}+\Lambda_{i}^{\prime2})\nonumber \\
 & +\frac{\hbar^{2}}{m}\sum_{i<j}[\Gamma_{ij}^{\prime\prime}+\Gamma_{ij}^{\prime2}+(\Lambda_{i}-\Lambda_{j})\Gamma_{ij}^{\prime}]\nonumber \\
 & -\frac{\hbar^{2}}{m}\sum_{\substack{i<j<k}
}(\Gamma_{ij}^{\prime}\Gamma_{jk}^{\prime}+\Gamma_{ij}^{\prime}\Gamma_{ki}^{\prime}+\Gamma_{ki}^{\prime}\Gamma_{jk}^{\prime}),\label{eq:H0-prime}
\end{align}
which is the same as the parent Hamiltonian in the real-valued TIJA,
except that $\mathscr{\hat{H}}_{0}^{\prime}(t)$ is now time-dependent.

We conclude this section by noting that the trial wave functions~(\ref{eq:TD-Jastrow})
and ~(\ref{eq:TDIR-Jastrow}) must be normalizable so that $\Psi$
and $\Phi$ do not blow up when particles are far apart.

\section{\label{sec:Imposing-Hermicitiy}Consistency conditions by imposing
Hermicity }

As already advanced, the Hermicity of $\mathscr{\hat{H}}(t)$ may
introduce strong constraints on the functional forms of $f_{ij}(t)$
and $g_{k}(t)$. To find such constraints, in principle, one can rewrite
$\hat{\mathscr{H}}(t)$ in terms of $\Gamma_{ij}(t)$, $\theta_{ij}(t)$,
$\Lambda_{k}(t)$ and $\phi_{k}(t)$ by substituting Eq.~(\ref{eq:complex-trial-WF})
into Eq.~(\ref{eq:scrH-def}) and then impose that the imaginary
part of $V$ vanishes. Specifically, we note that the two-body and
three-body terms can by no means be reduced to a one-body potential
unless they are independent of the particles' positions. This dictates
that $\text{Im}v_{1}^{(i)}$ must be a function of time only. We shall
assume similar constraints on the two-body and three-body interactions,
i.e., that $\text{Im}v_{2}^{(ij)}$ and $\text{Im}v_{3}^{(ijk)}$
are functions of time only. This assumption ignores the fact that
in some cases the three-body interaction $v_{3}^{(ijk)}$ can reduce
to two-body interactions~\citep{Calogero75,yang2022one} for the
sake of simplicity. Based on the above analysis, we introduce 
\begin{equation}
\dot{\widetilde{\mathcal{N}}}_{s}(t)=\text{Im}v_{s}^{(i_{1}\cdots i_{s})},\quad s=1,\,2,\,3,\,\label{eq:cal-Ns}
\end{equation}
where $\widetilde{\mathcal{N}}_{s}(t)$ is a real-valued function of time
and again, the over-dot denotes time derivation. Thus, the Hermiticity
of $\mathscr{\hat{H}}(t)$ imposes 
\begin{equation}
\mathcal{N}(t)\sim\frac{N\hbar}{2m}\widetilde{\mathcal{N}}_{1}(t)+\frac{\hbar N(N-1)}{2m}\widetilde{\mathcal{N}}_{2}(t)+\frac{\hbar N(N-1)(N-2)}{6m}\widetilde{\mathcal{N}}_{3}(t),
\end{equation}
where $\sim$ denotes equivalence up to a constant independent of
time and particles' positions. This is our first result for considering
the complex-valued TDJA. We will exemplify these results in several
examples.

In general, Eq.~(\ref{eq:cal-Ns}) for $s=3$ can lead to complicated
conditions. In this work, we shall restrict our attention to the case
where the three-body interactions are either real-valued, depending
on both coordinates and time or complex-valued and depending on time
only. \textcolor{blue}{The former assumes a vanishing two-body phase angle
$\theta_{ij}(t)$, } while the latter essentially boils the trial function
down to three types of functions 
\begin{equation}
e^{\Gamma_{ij}(t)}\propto\begin{cases}
|x_{ij}|^{\lambda(t)} & \text{CS}\\
|\sinh[c(t)x_{ij}]|^{\lambda(t)} & \text{Hyperbolic}\\
\exp[c(t)|x_{ij}|] & \text{LL}
\end{cases},\label{eq:trial-WF-three}
\end{equation}
corresponding to the celebrated CS, hyperbolic, and LL models, respectively
in the case of TDJA~\citep{delcampo20,Sutherland04}. In all these
cases, the following quantity 
\begin{equation}
W_{3}(\bm{x},\,t)\equiv-\sum_{\substack{i<j<k}
}(\Gamma_{ij}^{\prime}\Gamma_{jk}^{\prime}+\Gamma_{ij}^{\prime}\Gamma_{ki}^{\prime}+\Gamma_{ki}^{\prime}\Gamma_{jk}^{\prime})\label{eq:v3C-def}
\end{equation}
is independent of the particles' coordinates. \textcolor{blue}{Note that care needs to be taken when calculating the derivatives of $\Gamma_{ij}(t)$ according to Eq.~\eqref{eq:trial-WF-three} since it involves absolute value of a function, see e.g., Chapter 5 of Ref.~\cite{Sutherland04}and Supplemental Material of Ref.~\cite{yang2022one} and for details.}  In particular, one can
find that $W_{3}(t)$ reduces to a constant for the trial
wave functions in Eq.~(\ref{eq:trial-WF-three}), i.e., 
\begin{equation}
W_{3}(t)=\begin{cases}
0 & \text{CS}\\
N(N-1)(N-2)\lambda^{2}(t)c^{2}(t)/6 & \text{Hyperbolic}\\
N(N-1)(N-2)c^{2}(t)/6 & \text{LL}
\end{cases}.
\end{equation}
In Appendix~\ref{sec:Hermicitiy}, we show that the two-body phase
angle can be written in a unified expression, 
\begin{equation}
\theta_{ij}(t)=\eta(t)\Gamma_{ij}(t),\,f_{ij}(t)=e^{\Gamma_{ij}(t)[1+\eta(t)]},\label{eq:2-body-PA}
\end{equation}
where $\eta(t)$ is any given real-valued function of time. Then the
three-body term defined in Eq.~(\ref{eq:v3-def}) becomes 
\begin{align}
\sum_{\substack{i<j<k}
}v_{3}^{(ijk)} & =W_{3}^{\text{}}(t)[1+\text{i}\eta(t)]^{2}.\label{eq:U3-special}
\end{align}
Combining Eq.~(\ref{eq:U3-special}) with Eq.~(\ref{eq:cal-Ns}),
one finds $N(N-1)(N-2)\dot{\widetilde{\mathcal{N}}}_{3}(t)/6=2\eta(t)W_{3}(t).$
Note that when $\eta(t)=0$, $\dot{\mathcal{N}}_{3}(t)=0$ but $W_{3}(t)$
may depend on the coordinates in general. We can combine $\dot{\widetilde{\mathcal{N}}}_{2}(t)$
and $\dot{\widetilde{\mathcal{N}}}_{3}(t)$ and obtain

\begin{equation}
\textcolor{blue}{\mathcal{N}(t)\sim\mathcal{N}_{1}(t)+\mathcal{N}_{23}(t)},\label{eq:calN}
\end{equation}
where 
\begin{align}
\mathcal{N}_{1}(t) & \equiv\frac{N\hbar}{2m}\widetilde{\mathcal{N}}_{1}(t),\\
\mathcal{N}_{23}(t) & \equiv\frac{\hbar N(N-1)}{2m}\widetilde{\mathcal{N}}_{2}(t)+\frac{2\hbar}{m}\int_{0}^{t}\eta(\tau)W_{3}(\tau)d\tau.
\end{align}

Next, we note that in general Eq.~(\ref{eq:cal-Ns}) with $s=1,\,2$
yields for these models consistency conditions between the trial functions,
which are discussed in detail in Appendix~\ref{sec:Hermicitiy}.
\textcolor{blue}{Throughout this work, we shall focus on the case
where the particles are embedded in the harmonic trap, i.e., $\text{Re}v_{1}^{(i)}$
is quadratic in $x_{i}$ and therefore  $\Lambda_{k}(t)=-m\omega(t)x_{k}^{2}/(2\hbar)$.
As argued in Appendix~\ref{sec:Consistency-condition-Harmonic},
the one-body consistency condition, in this case, simplifies dramatically.
It explicitly gives the functional
form of $\mathscr{C}_{2}(t)$, i.e., 
\begin{equation}
\mathscr{C}_{2}(t)\equiv\eta(t)\omega(t)+\frac{\dot{\omega}(t)}{2\omega(t)},\label{eq:scrC-def}
\end{equation}
with 
\begin{align}
\dot{\widetilde{\mathcal{N}}}_{1}(t) & =-\frac{m\dot{\omega}(t)}{2\hbar\omega(t)},\\
\phi_{k}(t) & =-\frac{m\dot{\omega}(t)}{4\hbar\omega(t)}x_{k}^{2}+\tau(t).\label{eq:1-body-PA-harmonic}
\end{align}
The two-body consistency condition corresponding to $s=2$ is 
\begin{equation}
\eta(t)\Gamma_{ij}^{\prime\prime}+2\eta(t)\Gamma_{ij}^{\prime2}+\frac{m}{\hbar}\dot{\Gamma}_{ij}-\frac{m}{\hbar}\mathscr{C}_{2}(t)\Gamma_{ij}^{\prime}x_{ij}=\dot{\widetilde{\mathcal{N}}}_{2}(t).\label{eq:2-body-consist-main}
\end{equation}
} The Hamiltonian (\ref{eq:scrH-def}) then becomes 
\begin{align}
\hat{\mathscr{H}}(t) & =\frac{1}{2m}\sum_{i}p_{i}^{2}+\frac{1}{2}m\Omega^{2}(t)\sum_{i}x_{i}^{2}\nonumber \\
 & +\frac{\hbar^{2}}{m}\sum_{i<j}\left(\Gamma_{ij}^{\prime\prime}+[1-\eta^{2}(t)]\Gamma_{ij}^{\prime2}\right)\nonumber \\
 & -\hbar\varpi(t)\sum_{i<j}\Gamma_{ij}^{\prime}x_{ij}-\hbar\sum_{i<j}\frac{d}{dt}[\eta(t)\Gamma_{ij}]+\mathscr{E}(t),\label{eq:TDPHJ}
\end{align}
where the frequency of the trap $\Omega(t)$ is defined as 
\begin{equation}
\Omega^{2}(t)\equiv\omega^{2}(t)+\frac{d}{dt}\left[\frac{\dot{\omega}(t)}{2\omega(t)}\right]-\left[\frac{\dot{\omega}(t)}{2\omega(t)}\right]^{2},\label{eq:omg-def}
\end{equation}
and 
\begin{align}
\mathscr{E}(t) & \equiv-\frac{1}{2}N\hbar\vartheta(t)+\frac{\hbar^{2}}{m}[1-\eta^{2}(t)]W_{3}(t),\label{eq:scrE-def}\\
\varpi(t) & \equiv\omega(t)\left[1-\frac{\eta(t)\dot{\omega}(t)}{2\omega^{2}(t)}\right],\label{eq:omgbar-def}\\
\vartheta(t) & \equiv2\dot{\tau}(t)+\omega(t).\label{eq:vartheta-def}
\end{align}

Upon making the change of variables 
\begin{equation}
\omega(t)=\frac{\omega_{0}}{b^{2}(t)},\,\omega_{0}\equiv\omega(0),
\end{equation}
one immediately observes that $\dot{\omega}(t)/[2\omega(t)]=-\dot{b}(t)/b(t)$
and Eq.~(\ref{eq:omg-def}) becomes 
\begin{equation}
\ddot{b}(t)+\Omega^{2}(t)b(t)=\omega_{0}^{2}b^{-3}(t),\label{eq:Ermakov}
\end{equation}
with the initial condition $b(0)=1$. Remarkably, Eq. (\ref{eq:Ermakov})
is the celebrated Ermakov equation governing the dynamics of the scaling
factor in scale-invariant quantum many-body evolution, emerging when
the interactions of particles have given scaling properties~\citep{delcampo16,gritsev2010scaling,Castin04}.
However, as we shall discuss in Sec.~\ref{sec:Effective-dynamics},
although the effective scale-invariant dynamics also appears here,
the interactions need not be restricted to a given scaling dimension.

Let us remark on some differences with previous approaches in the literature~\cite {delcampo16,gritsev2010scaling,Castin04}.
(i) Works on scale invariance consider  $\dot{b}(0)=0$ when the initial state is stationary.
Here we only impose the condition $b(0)=1$ for the Ermakov equation, and $\dot{b}(0)$ can be arbitrary. (ii) In the Ermakov
equation discussed by the previous literature, it is assumed that
$\Omega_{0}\equiv\Omega(0)$ is always equal to $\omega_{0}\equiv\omega(0)$.
Here, such a constraint does not necessarily hold. As we shall discuss
in detail in Sec.~\ref{sec:Effective-dynamics}, if we further impose
$\dot{b}(0)=0$ and $\Omega_{0}=\omega_{0}$ together with the conditions
for $\eta(t)$ at $t=0$, we obtain $\hat{\mathscr{H}}(0)\Psi(0)=0$,
i.e., at time $t=0$, the TDJA is also an eigenstate of the corresponding
parent Hamiltonian, which is always a requirement for finding the
counterdiabatic driving for scale-invariant dynamics, as discussed
previously \citep{delcampo2013shortcuts,deffner2014classical}. We
emphasize that in our discussion, the initial time $t_{0}$ is
not necessarily zero. In fact, both Eq.~(\ref{eq:Psi-TDJastrow})
and the interaction in Eq.~(\ref{eq:TDPHJ}) can break, in general,
scale-invariance. As we shall see subsequently, both scale-invariant
and non-scale-invariant dynamics follow naturally from our results.

\textcolor{blue}{Finally, we shall refer to $b(t)$ along with $\eta(t)$ and
$\gamma(t)$ [defined in Sec.~\ref{subsec:The-logarithmic-hyperbolic-model}],
as the fundamental  parameters that govern the exact many-body dynamics and regard other time-dependent
parameters as secondary given that they can be derived from the former. Throughout this work, while the
time-dependent parent Hamiltonian~(\ref{eq:TDPHJ}) may contain both
types of parameters, we shall express the generic time-dependent Jastrow
ansatz~(\ref{eq:TD-Jastrow}) mainly in terms of the fundamental
parameters. To this goal, Eq.~(\ref{eq:TD-Jastrow}) can be rewritten
as 
\begin{equation}
\Psi(\bm{x},\,t)=\frac{1}{e^{\mathcal{N}_{23}(t)}b^{N/2}(t)}\prod_{i<j}e^{\Gamma_{ij}(t)[1+\text{i}\eta(t)]}\prod_{k}e^{-\frac{x_{k}^{2}}{2x_{0}^{2}b^{2}(t)}+\text{i}\frac{\dot{b}(t)x_{k}^{2}}{2\omega_{0}x_{0}^{2}b(t)}+\text{i}\tau(t)},\label{eq:Psi-TDJastrow}
\end{equation}
where we have ignore the constant factor $\omega_{0}^{N/4}$ and $x_{0}$
is the length scale of the harmonic trap defined as 
$x_{0}\equiv\sqrt{\frac{\hbar}{m\omega_{0}}}$.
}
\textcolor{blue}{To summarize:
given
the functions $\Gamma_{ij}(t)$ and $\eta(t)$ satisfy the consistency
condition~(\ref{eq:2-body-consist-main}),  the nonequilibrium dynamics of one-dimensional interacting bosons in a harmonic trap described by associated with the Hamiltonian~(\ref{eq:Psi-TDJastrow})
admits 
an exact description  in terms of  Eq.~(\ref{eq:TDPHJ}), which bears
the form of TDJA. Such nonequilibrium dynamics are exemplified by
specifying the functional form of $\Gamma_{ij}(t)$ and $\eta(t)$
in Sec.~\ref{sec:vanishing-eta}-\ref{sec:TD-long-range-LL}.
}Knowledge of the exact nonequilibrium dynamics of strongly-correlated
quantum systems is rare and precious, and we shall explore its applications
to counterdiabatic driving in STA~\citep{demirplak2003adiabatic,demirplak2005assisted,demirplak2008onthe,berry2009transitionless,Torrontegui13,guery-odelin2019shortcuts}
and quench dynamics subsequently.

\section{Shortcuts to adiabaticity for the Jastrow ansatz}

In a general setting, for a given reference Hamiltonian as $\hat{\mathscr{H}}_{0}^{\prime}(t)$,
with instantaneous eigenvectors $\Phi_{n}(t)$ and eigenvalues $\varepsilon_{n}(t)$,
one can consider the use of counterdiabatic control fields $\hat{\mathscr{H}}_{1}^{\prime}(t)$
that are able to drive the dynamics through the adiabatic manifold
of $\Phi_{n}(t)$ parametrized by $t$ \citep{demirplak2003adiabatic,demirplak2005assisted,demirplak2008onthe,berry2009transitionless}.
This is the goal of counterdiabatic driving, also known as transitionless
quantum driving, a technique that provides a universal way to construct
STA protocols \citep{Torrontegui13,guery-odelin2019shortcuts}. The
application of these techniques to quantum fluids has led to manifold
applications ranging from quantum microscopy \citep{delcampo11,Papoular15}
to the engineering of efficient friction-free quantum thermal machines
\citep{Beau16,Deng18Sci,delCampo2018}. Experimental progress has
focused on the case of atomic clouds driven by time-dependent confinements
\citep{Schaff10,Schaff11,Schaff11njp,Rohringer2015,Deng18,Deng18Sci}.

Counterdiabatic driving aims at finding a global Hamiltonian 
\begin{equation}
\hat{\mathscr{H}}^{\prime}(t)=\hat{\mathscr{H}}_{0}^{\prime}(t)+\hat{\mathscr{H}}_{1}^{\prime}(t),\label{eq:Hprime-STA}
\end{equation}
where 
\begin{equation}
\mathscr{\hat{H}}_{0}^{\prime}(t)\Phi_{n}(t)=\varepsilon_{n}(t)\Phi_{n}(t).\label{eq:Hprime-STA-def}
\end{equation}

In general, finding the Hamiltonian $\mathscr{H}^{\prime}(t)$ that satisfies
Eq.~(\ref{eq:Hprime-STA-def}) for a specific eigenstate may not
be easy. Nevertheless, if Eq.~(\ref{eq:Hprime-STA-def}) is valid
for all eigenstates, it can be shown that $\hat{\mathscr{H}}^{\prime}(t)=\text{i}\hbar\sum_{n}\ket{\dot{\Phi}_{n}(t)}\bra{\Phi_{n}(t)}$,
so that 
\begin{equation}
\mathscr{\hat{H}}_{1}^{\prime}(t)=\sum_{n}\left[\text{i}\hbar\ket{\dot{\Phi}_{n}(t)}\bra{\Phi_{n}(t)}-\varepsilon_{n}\ket{\Phi_{n}(t)}\bra{\Phi_{n}(t)}\right].\label{eq:H1prime-exp}
\end{equation}

Known explicit constructions of $\hat{\mathscr{H}}^{\prime}(t)$
are  restricted to particular cases, including harmonic oscillator~\citep{Muga2010},
scale-invariant dynamics~\citep{delcampo2013shortcuts,deffner2014classical},
integrable systems~\citep{Okuyama16}, etc. In the engineering of
STA for many-body systems, several difficulties may occur: (i) Computing
$\hat{\mathscr{H}}_{1}^{\prime}(t)$ requires knowledge of the spectrum
of $\mathscr{\hat{H}}_{0}^{\prime}$,which is difficult to compute. 
(ii) The potential in $\hat{\mathscr{H}}_{1}^{\prime}(t)$ usually
contain nonlocal terms, i.e., terms that combine both the position
and momentum operators and cannot be associated with an external local
potential or interaction. For scale-invariant invariant dynamics,
it was shown that~\citep{delcampo2013shortcuts,deffner2014classical}
these two difficulties can be removed, and the non-locality of the
potential in $\mathscr{\hat{H}}_{1}^{\prime}(t)$ can be gauged away.

However, there is little understanding of the exact counterdiabatic
driving in many-particle systems beyond the scale-invariant dynamics.
To the best of our knowledge, there is not yet an analytical finding
of the counterdiabatic driving that is only valid for a specific
eigenstate. In what follows, we will show that the parent Hamiltonian
of the complex-valued TDJA Eq.~(\ref{eq:TDPHJ}) can be employed
to find the counterdiabatic driving of the prime parent Hamiltonian
of the real-valued TDJA defined in Eq.~(\ref{eq:H0-prime}). As is
clear from Eq.~(\ref{eq:Psi-TDJastrow}), the dynamics is not necessarily
scale-invariant. Furthermore, the construction does not rely on (\ref{eq:H1prime-exp})
and hence provides an example where the construction of an exact counterdiabatic
Hamiltonian for a given eigenstate does not necessarily drive the whole adiabatic manifolds
associated with all the energy eigenstates. Indeed, in stark contrast with the case of scale invariance, 
such construction is simplified by focusing the protocol on a given many-body state.

As already mentioned in the last section, we shall focus on the case
of a harmonic trap with either real-valued two-body trial wave function,
i.e., $\theta_{ij}(t)=0$ or $\Gamma_{ij}$ being \textcolor{blue}{given by Eq.~(\ref{eq:trial-WF-three})}. The
 many-body phase operator and the real-valued TDJA become 
\begin{align}
\mathscr{U}_{P}(\bm{x},\,t) & =\prod_{i<j}e^{\text{i}\eta(t)\Gamma_{ij}(t)}\prod_{k}e^{\text{i}\frac{\dot{b}(t)x_{k}^{2}}{2\omega_{0}x_{0}^{2}b(t)}+\text{i}\tau(t)},\label{eq:scr-U-harmonic}\\
\Phi(\bm{x},\,t) & =\frac{1}{e^{\mathcal{N}_{23}(t)}b^{N/2}(t)}\prod_{i<j}e^{\Gamma_{ij}(t)}\prod_{k}e^{-\frac{x_{k}^{2}}{2x_{0}^{2}b^{2}(t)}}.\label{eq:Phi-def-harmonic}
\end{align}
The parent Hamiltonian of the real-valued TIJA in Eq.~(\ref{eq:H0-prime})
reduces to 
\begin{align}
\hat{\mathscr{H}}_{0}^{\prime}(t) & =\frac{1}{2m}\sum_{i}p_{i}^{2}+\frac{1}{2}m\omega^{2}(t)\sum_{i}x_{i}^{2}+\frac{\hbar^{2}}{m}\sum_{i<j}\left(\Gamma_{ij}^{\prime\prime}+\Gamma_{ij}^{\prime2}\right)\nonumber \\
 & -\hbar\omega(t)\sum_{i<j}\Gamma_{ij}^{\prime}x_{ij}+\mathscr{E}_{0}^{\prime}(t),\label{eq:H0prime-CD}
\end{align}
where 
\begin{equation}
\mathscr{E}_{0}^{\prime}(t)\equiv-\frac{1}{2}N\hbar\omega(t)+\frac{\hbar^{2}}{m}W_{3}(t).
\end{equation}

Since $\Phi(t)$ is the instantaneous eigenstate of $\hat{\mathscr{H}}_{0}^{\prime}(t)$,
keeping the adiabatic evolution along $\Phi(t)$ requires infinitely
slow driving of $\hat{\mathscr{H}}_{0}^{\prime}(t)$. Using Eq.~(\ref{eq:U-connection}),
we find the counterdiabatic driving $\hat{\mathscr{H}}^{\prime}(t)$
is (see details in Appendix~\ref{sec:PHJ-rotating}) 
\begin{align}
\hat{\mathscr{H}}^{\prime}(t) & =\frac{1}{2m}\sum_{i}p_{i}^{2}+\frac{1}{2}m\omega^{2}(t)\sum_{i}x_{i}^{2}+\frac{\hbar^{2}}{m}\sum_{i<j}\left(\Gamma_{ij}^{\prime\prime}+\Gamma_{ij}^{\prime2}\right)\nonumber \\
 & -\hbar\omega(t)\sum_{i<j}\Gamma_{ij}^{\prime}x_{ij}+\mathscr{E}_{0}^{\prime}(t)+\hat{\mathscr{H}}_{1}^{\prime}(t),\label{eq:Hprime-CD}
\end{align}
where $\mathscr{\hat{H}}_{1}^{\prime}(t)$ contains the non-local one and two-body 
terms 
\begin{equation}
\hat{\mathscr{H}}_{1}^{\prime}(t)\equiv\frac{\dot{b}(t)}{2b(t)}\sum_{i}\{x_{i},\,p_{i}\}+\frac{\hbar\eta(t)}{2m}\sum_{i<j}\{\Gamma_{ij}^{\prime}(t),\,p_{ij}\},\label{eq:H1prime-def}
\end{equation}
and $p_{ij}\equiv p_{i}-p_{j}$. The counterdiabatic control term
$\mathscr{\hat{H}}_{1}^{\prime}(t)$ is non-local in the sense that
it contains the sum of the one-body squeezing operator, familiar from
the study of driven of scale-invariant systems \citep{Ibanez12,Jarzynski13,delcampo2013shortcuts,deffner2014classical},
and its two-body generalization, involving $\{\Gamma_{ij}^{\prime}(t),\,p_{ij}\}$.

Although the momentum dependence of $\hat{\mathscr{H}}^{\prime}(t)$ can make it difficult
to implement in a given platform, we note that the parent Hamiltonian of the complex-valued
TDJA $\hat{\mathscr{H}}(t)$ provides a feasible alternative, being local as long
as the many-body phase operation $\mathscr{U}_{P}(\bm{x},\,t)$ can
be implemented. Furthermore, it is straightforward to check that when
$\eta(t)=0$, $\hat{\mathscr{H}}_{0}^{\prime}(t)$ bears the same
form as $\hat{\mathscr{H}}(t)$, except for the fact that the driving frequency
in the former is $\omega(t)$ while in the latter it is $\Omega(t)$, and that they also differ in the expressions for the time-dependent constants $\mathscr{E}_{0}^{\prime}(t)$
and $\mathscr{E}(t)$. Similar observations are discussed in Ref.~\citep{delcampo2013shortcuts,deffner2014classical},
where the interactions have scaling properties. However, our approach lifts such a restriction.

In what follows, we shall apply the analysis in the preceding sections to describe the
exact dynamics of these one-dimensional strongly-correlated systems. To this end, we
discuss specifically the case of vanishing $\eta(t)$ and the CS, hyperbolic, and LL models.

\section{\label{sec:vanishing-eta}The case of vanishing $\eta(t)$}

As we have discussed before, when $\eta(t)=0$, the trial wave function
is not necessarily restricted to Eq.~(\ref{eq:trial-WF-three}),
but can be generic as long as it respects the bosonic symmetry. In
this case, $W_{3}(t)$ is a three-body potential that depends on both
coordinates and time and $\dot{\widetilde{\mathcal{N}}}_{3}(t)=0$. The
two-body consistency condition~(\ref{eq:2-body-consist-main}) simplifies
to 
\begin{equation}
\frac{m}{\hbar}\dot{\Gamma}_{ij}+\frac{\dot{b}(t)}{b(t)}\frac{m}{\hbar}\Gamma_{ij}^{\prime}x_{ij}=\dot{\widetilde{\mathcal{N}}}_{2}(t).\label{eq:consistency-vanishing-eta}
\end{equation}
This condition can be satisfied if  $\Gamma_{ij}(t)$ takes
the form 
\begin{equation}
\Gamma_{ij}(t)=\Gamma\left(\frac{c_{0}x_{ij}}{b(t)}\right), \label{eq:vanishing-eta-Gamma}
\end{equation}
up to a constant with $\dot{\widetilde{\mathcal{N}}}_{2}(t)=0$. As a
result, Eq.~(\ref{eq:TDPHJ}) becomes 
\begin{align}
\hat{\mathscr{H}}(t) & =\frac{1}{2m}\sum_{i}p_{i}^{2}+\frac{1}{2}m\Omega^{2}(t)\sum_{i}x_{i}^{2}+\frac{\hbar^{2}}{m}\sum_{i<j}\left(\Gamma_{ij}^{\prime\prime}+\Gamma_{ij}^{\prime2}\right)\nonumber \\
 & -\hbar\omega(t)\sum_{i<j}\Gamma_{ij}^{\prime}x_{ij}+\frac{\hbar^{2}}{m}W_{3}(\bm{x},\,t)-\frac{1}{2}N\hbar\vartheta(t).\label{eq:H-vanishing-eta}
\end{align}
Note that the three-body interaction $W_{3}(\bm{x},\,t)$ defined
in Eq.~(\ref{eq:v3C-def}) depends on time and the coordinates
in this case. The dynamics is given by 
\begin{equation}
\Psi(\bm{x},\,t)=\frac{1}{[b(t)]^{N/2}}\prod_{i<j}e^{\Gamma\left(\frac{c_{0}x_{ij}}{b(t)}\right)}\prod_{k}e^{-\frac{x_{k}^{2}}{2x_{0}^{2}b^{2}(t)}+\text{i}\frac{\dot{b}(t)x_{k}^{2}}{2\omega_{0}x_{0}^{2}b(t)}+\text{i}\tau(t)},\label{eq:Psi-vanishing-eta}
\end{equation}
which is scale-invariant. However, we note that the interactions of
the particles do not necessarily have scaling properties, as we will
see in the subsequent sections when $\Gamma_{ij}(t)$ takes any of
the forms in Eq.~(\ref{eq:trial-WF-three}).

The prime parent Hamiltonian~(\ref{eq:H0prime-CD}) is the same
as Eq.~(\ref{eq:H-vanishing-eta}), but for the replacement of $\Omega(t)$
with $\omega(t)$. The auxiliary driving now takes a simple form 
\begin{equation}
\hat{\mathscr{H}}_{1}^{\prime}(t)=\frac{\dot{b}(t)}{2b(t)}\sum_{i}\{x_{i},\,p_{i}\}.\label{eq:control-CSA}
\end{equation}
It is worth noting that this is the familiar counterdiabatic term
previously discussed for STA governed by scale-invariance in Refs.~\citep{delcampo2013shortcuts,deffner2014classical,Beau16},
where the interparticle interactions have the same scaling dimension
as the kinetic energy \citep{gritsev2010scaling,delcampo11}.

Finally, we conclude this section by noting that with the complex-valued
TDJA~(\ref{eq:Psi-vanishing-eta}), one can compute the survival
probability defined as 
\begin{align}
\text{Pr}_{s}(t,\,t_{0}) & \equiv\frac{1}{\mathcal{C}^{2}}|\braket{\Psi(t)\big|\Psi(t_{0})}|^{2},\label{eq:SP-def}
\end{align}
where $t_{0}$ is the initial time, which is not necessarily set to
be zero, as further discussed in Sec.~\ref{sec:Quench-dynamics}
and $\mathcal{C}$ is the normalization of the TDJA
\begin{equation}
\mathcal{C}\equiv\braket{\Psi(t)\big|\Psi(t)}=\int_{\mathbb{R}^N}d^Nx\Psi(\bm{x},t)\Psi^{*}(\bm{x},\,t)\label{eq:calC-def}
\end{equation}
The evaluation of $\braket{\Psi(t)\big|\Psi(t_{0})}$ and $\mathcal{C}$
involves the following multi-dimensional integrals \begin{widetext}
\begin{equation}
\braket{\Psi(t)\big|\Psi(t_{0})}=\frac{e^{-\text{i}N[\tau(t)-\tau(t_{0})]}}{[b(t)b(t_{0})]^{N/2}}\int_{\mathbb{R}^N}d^Nx\exp\left(\sum_{i<j}\left[\Gamma\left(\frac{c_{0}x_{ij}}{b(t)}\right)+\Gamma\left(\frac{c_{0}x_{ij}}{b(t_{0})}\right)\right]-\frac{1}{2x_{0}^{2}}\sum_{k}x_{k}^{2}\left[\frac{1}{b^{2}(t)}+\frac{1}{b^{2}(t_{0})}-\frac{\text{i}}{\omega_{0}}\left(\frac{\dot{b}(t_{0})}{b(t_{0})}-\frac{\dot{b}(t)}{b(t)}\right)\right]\right).\label{eq:unnorm-amplitude}
\end{equation}
One can further make the change of variables $\alpha(t,\,t_{0})x_{k}=z_{k}$,
where 
\begin{equation}
\alpha^{2}(t,\,t_{0})\equiv\frac{1}{2}\left[\frac{1}{b^{2}(t)}+\frac{1}{b^{2}(t_{0})}-\frac{\text{i}}{\omega_{0}}\left(\frac{\dot{b}(t_{0})}{b(t_{0})}-\frac{\dot{b}(t)}{b(t)}\right)\right],
\end{equation}
and $\alpha(t_{0},\,t_{0})=\frac{1}{b(t_{0})}$. As a result, Eq.~(\ref{eq:unnorm-amplitude})
becomes
\begin{equation}
\braket{\Psi(t)\big|\Psi(t_{0})}=\frac{e^{-\text{i}N[\tau(t)-\tau(t_{0})]}}{[b(t)b(t_{0})]^{N/2}\alpha^{N}(t,\,t_{0})}\int_{\mathscr{C}^N}d^Nz\exp\left(\sum_{i<j}\left[\Gamma\left(\frac{c_{0}}{b(t)}\frac{z_{ij}}{\alpha(t,\,t_{0})}\right)+\Gamma\left(\frac{c_{0}}{b(t_{0})}\frac{z_{ij}}{\alpha(t,\,t_{0})}\right)\right]-\sum_{k}z_{k}^{2}/x_{0}^{2}\right),\label{eq:SA-zvar}
\end{equation}
\end{widetext}where the integral contour $\mathscr{C}$ is $(-\infty e^{\text{i}\text{arg}[\alpha(t,\,t_{0})]},\,\infty e^{\text{i\text{arg}[\ensuremath{\alpha}(t,\,\ensuremath{t_{0}})]}})$,
with $\text{arg}[\alpha(t,\,t_{0})|]$ denoting the argument of the
complex variable $\alpha(t,\,t_{0})$. 

\textcolor{blue}
{
The absolute value in the two-body wave function $\Gamma_{ij}$ or
$e^{\Gamma_{ij}}$ imposed by the bosonic symmetry {[}see e.g. Eq.~(\ref{eq:trial-WF-three}){]}
renders the integrand of Eq.~(\ref{eq:SA-zvar}) non-analytic on
the complex plane in general. Therefore, one \textit{cannot}
change integral contour from $\mathscr{C}$ to $[-\infty,\,\infty]$.
}

\textcolor{blue}{
Nevertheless, for the time-dependent Calogero-Sutherland model discussed in Sec.~\ref{subsec:CS-Scale-Invariant}, thanks to the scaling property $e^{\Gamma(\lambda x_{ij})}=\lambda e^{\Gamma(x_{ij})}$ in this case, the integrand of Eq.~(\ref{eq:SA-zvar}) is analytic [see Eq.~\eqref{eq:CS-amplitude-int}]. We note that the survival probability for the scale-invariant
dynamics in the CSM and in interacting ultracold
gases in the unitary limit have been computed in Refs.~\citep{delcampo16,delcampo21},
implicitly assuming the analytic continuation mentioned above.
}

\textcolor{blue}
{
On the other hand, for the
hyperbolic model and the long-range Lieb Linger model discussed in Sec.~\ref{sec:The-time-dependent-Hyperbolic} and Sec.~\ref{sec:TD-long-range-LL} respectively, 
one can estimate the long-time asymptotic behavior of
the survival probability through Eq.~(\ref{eq:SA-zvar}), Eq.~(\ref{eq:SA-zvar}) cannot be analytically continued to the real line but remain bounded. As a result, we can estimate the scaling behavior of the survival probability with respect to the number of particles. }

\section{\label{sec:TD-CS}The time-dependent Calogero-Sutherland models}

The CS model has many applications in different fields. As
we have discussed, it describes ultracold atoms confined in tight
waveguides in the Tonks-Girardeau limit and its generalizations with
inverse-square interactions. It also appears naturally in matrix models,
random matrix theory, and quantum chaos \citep{Forrester10,Haake,Marino15},
even if the system is itself integrable. As we show next, our framework
accounts for the dynamics of the CS model and its generalizations.

We start with the generic time-dependent Jastrow ansatz without assuming
scale invariance by taking $\Gamma_{ij}(t)=\lambda(t)\ln|x_{ij}|$
in Eq.~(\ref{eq:Psi-TDJastrow}) so that $\Gamma_{ij}^{\prime}=\lambda(t)/x_{ij}$.
The three-body term $W_{3}(\bm{x},\,t)$ vanishes in this case regardless
of the value of $\eta(t)$. The two-body consistency equation~(\ref{eq:2-body-consist-main})
yields two nontrivial cases~\footnote{The case $\lambda(t)=0$ is trivial}:

\begin{enumerate}[label={\Alph*)}]

\item When $\eta(t)=0$ and $\lambda(t)=\lambda_{0}$, with $\lambda_{0}$
being an arbitrary time-independent real number, $\widetilde{\mathcal{N}}_{3}(t)=0$,
and 
\begin{equation}
\dot{\widetilde{\mathcal{N}}}_{2}(t)=-\frac{m\dot{\omega}(t)\lambda_{0}}{2\hbar\omega(t)},\:\dot{\mathcal{N}}_{23}(t)=\frac{\dot{b}(t)\lambda_{0}N(N-1)}{2b(t)}.\label{eq:N23-dot}
\end{equation}

\item When $\lambda(t)=1/2$ with arbitrary $\eta(t)$, $\widetilde{\mathcal{N}}_{3}(t)=0$.
$\dot{\widetilde{\mathcal{N}}}_{2}(t)$ and $\dot{\mathcal{N}}_{23}(t)$
are still given by Eq.~(\ref{eq:N23-dot})  with $\lambda_{0}=1/2$.

\end{enumerate}

One can observe that shifting $\Gamma_{ij}(t)$ by a
time-dependent function preserves the condition in Eq.~(\ref{eq:2-body-consist-main}).
It only shifts the $\dot{\mathcal{N}}_{2}(t)$ by the same time-dependent
function. \textcolor{blue}{Thus, for case A ,  by replacing  $\Gamma_{ij}(t)=\lambda_0 \ln|x_{ij}|$ with  $\lambda_0(\ln |x_{ij}|-\ln b(t))$,  the equation ~\eqref{eq:consistency-vanishing-eta} is satisfied with $\dot{\mathcal{\tilde{N}}}_{2}(t)=0$ for the case of vanishing $\eta(t)$.}
We now discuss these two cases respectively.

\subsection{The Calogero-Sutherland model: Recovering the scaling-invariant dynamics\label{subsec:CS-Scale-Invariant}}

In the CS model, 
\begin{equation}
\mathcal{N}_{23}(t)=\frac{1}{2}\lambda_{0}N(N-1)\ln b(t).\label{eq:N-scale-invaraint}
\end{equation}
Eq.~(\ref{eq:Psi-TDJastrow}) becomes 
\begin{equation}
\Psi(\bm{x},\,t)=\frac{1}{[b(t)]^{N/2}}\prod_{i<j}\bigg|\frac{x_{ij}}{b(t)}\bigg|^{\lambda_{0}}e^{-\frac{\sum_{k}x_{k}^{2}}{2x_{0}^{2}b^{2}(t)}+\text{i}\frac{\dot{b}(t)\sum_{k}x_{k}^{2}}{2\omega_{0}x_{0}^{2}b(t)}+\text{i}N\tau(t)},\label{eq:Psi-CS-scale-invariant}
\end{equation}
which is the scale-invariant dynamics for the well-known Calogero-Sutherland
model in the literature~\citep{Sutherland98,DaeYupSong01,DaeYupSong02,delcampo16}.
Note that the norm of this many-body function is preserved in time
upon making the change of variables $x_{i}/b(t)\to x_{i}$.

Equation~(\ref{eq:TDPHJ}) becomes 
\begin{equation}
\hat{\mathscr{H}}(t)=\frac{1}{2m}\sum_{i}p_{i}^{2}+\frac{1}{2}m\Omega^{2}(t)\sum_{i}x_{i}^{2}+\frac{\hbar^{2}}{m}\sum_{i<j}\frac{\lambda_{0}[\lambda_{0}-1]}{x_{ij}^{2}}+\bar{\mathscr{E}}(t),\label{eq:TDPHJ-CS-regular}
\end{equation}
where 
\begin{equation}
\mathscr{\bar{E}}(t)=\mathscr{E}(t)-\frac{1}{2}N(N-1)\hbar\omega(t)\lambda_{0}.\label{eq:0-energy-CS}
\end{equation}
Upon making $\omega(t)$ time-independent and setting $\tau(t)=0$,
Eq.~(\ref{eq:0-energy-CS}) agrees with the one found from time-independent
Jastrow ansatz~\citep{delcampo20,Calogero71}. The position-independent
phase  $\tau(t)$ can be determined by imposing $\mathscr{\bar{E}}(t)=0$,
leading to 
\begin{equation}
\tau(t)=-\frac{1}{2}[(N-1)\lambda_{0}+1]\int^{t}\frac{\omega_{0}}{b^{2}(s)}ds.
\end{equation}

Finally, for the counterdiabatic driving $\hat{\mathscr{H}}^{\prime}(t)=\mathscr{\hat{H}}_{0}^{\prime}(t)+\hat{\mathscr{H}}_{1}^{\prime}(t)$,
the reference Hamiltonian Eq.~(\ref{eq:H0prime-CD}) becomes 
\begin{equation}
\hat{\mathscr{H}}_{0}^{\prime}(t)=\sum_{i}\frac{p_{i}^{2}}{2m}+\frac{1}{2}m\omega^{2}(t)\sum_{i}x_{i}^{2}+\frac{\hbar^{2}}{m}\sum_{i<j}\frac{\lambda_{0}(\lambda_{0}-1)}{x_{ij}^{2}}+\mathscr{\bar{E}}_{0}^{\prime}(t),\label{eq:H0prime-CS}
\end{equation}
where 
\begin{equation}
\mathscr{\bar{E}}_{0}^{\prime}(t)=\mathscr{E}_{0}^{\prime}(t)-\frac{1}{2}N(N-1)\hbar\omega(t)\lambda_{0}.
\end{equation}
The auxiliary control field is given by Eq.~(\ref{eq:control-CSA}).
We see that for counterdiabatic driving or the parent Hamiltonian of
the real-valued TIJA, Eq.~(\ref{eq:Hprime-CD}) is the same as the
one found in Refs.~\citep{delcampo2013shortcuts,deffner2014classical},
which takes the advantage of the scale-invariance of the inverse square
potential in the Calogero-Sutherland model.

While $\hat{\mathscr{H}}^{\prime}(t)$ involves a squeezing term, in  $\hat{\mathscr{H}}(t)$ the nonlocal terms
disappear. In this particular case, $\hat{\mathscr{H}}(t)$ is scale invariant
and hence $\hat{\mathscr{H}}^{\prime}(t)$ can keep the adiabatic
evolution of all the eigenstates of $\hat{\mathscr{H}}_{0}^{\prime}(t)$~\citep{delcampo2013shortcuts,deffner2014classical},
and not only that through the adiabatic trajectory prescribed by $\Phi(t)$.

The analytic expression for the survival probability for the CS model
was first obtained exactly in Ref.~\citep{delcampo16}. Here, we reproduce
it  carefully, addressing the issue of analytic continuation for
Eq.~ (\ref{eq:SA-zvar}). By taking advantage of the scaling property
of the two-body wave function, we can rewrite Eq.~ (\ref{eq:SA-zvar})
as\begin{widetext}
\begin{equation}
\braket{\Psi(t)\big|\Psi(t_{0})}=\frac{1}{[b(t)b(t_{0})]^{N/2}\alpha(t,\,t_{0}){}^{N}}\int_{\mathscr{C}^N}d^Nz\prod_{i<j}\frac{z_{ij}{}^{2\lambda_{0}}}{\alpha(t,\,t_{0}){}^{2\lambda_{0}}[b(t)b(t_{0})]^{\lambda_{0}}}e^{-\sum_{k}z_{k}^{2}/x_{0}^{2}}\label{eq:CS-amplitude-int}.
\end{equation}
\end{widetext}
\textcolor{blue}{The integrand is analytic in each variable $z_{k}$ and decays to zero exponentially as $|z_{k}|\to\infty$. Thus for each
single integration, one can analytically continue the integration domain
from $\mathscr{C}$ to $\mathbb{R}$.
Upon substituting the overlap integral into
(\ref{eq:SP-def}), one can find the multi-dimensional integrals in
the denominator and the numerator cancel with each other, and the survival
probability  takes the simple form 
\begin{equation}
\text{Pr}_{s}(t,\,t_{0})=\left[\frac{1}{b(t)b(t_{0})|\alpha(t,\,t_{0})|^{2}}\right]^{N[1+\lambda_{0}(N-1)]}.\label{eq:SP-CS}
\end{equation}
}

\subsection{The logarithmic Calogero-Sutherland model }

For case (ii), where $\Gamma_{ij}=\ln|x_{ij}|/2$, the complex-valued
TDJA~(\ref{eq:Psi-TDJastrow}) becomes 
\begin{equation}
\Psi(\bm{x},\,t)=\frac{1}{[b(t)]^{N/2}}\prod_{i<j}\bigg|\frac{x_{ij}}{b(t)}\bigg|^{1/2}|x_{ij}|^{\text{i}\eta(t)/2}e^{-\frac{\sum_{k}x_{k}^{2}}{2x_{0}^{2}b^{2}(t)}+\text{i}\frac{\dot{b}(t)\sum_{k}x_{k}^{2}}{2\omega_{0}x_{0}^{2}b(t)}+\text{i}N\tau(t)}.\label{eq:Psi-LogCS}
\end{equation}
The corresponding parent Hamiltonian is 
\begin{align}
\hat{\mathscr{H}}(t) & =\frac{1}{2m}\sum_{i}p_{i}^{2}+\frac{1}{2}m\Omega^{2}(t)\sum_{i}x_{i}^{2}+\bar{\mathscr{E}}(t)\nonumber \\
 & -\frac{\hbar^{2}}{m}\sum_{i<j}\frac{\eta^{2}(t)+1}{4x_{ij}^{2}}-\frac{\hbar\dot{\eta}(t)}{2}\sum_{i<j}\ln|x_{ij}|,\label{eq:TDPHJ-CS}
\end{align}
where 
\begin{equation}
\bar{\mathscr{E}}(t)=\mathscr{E}(t)-\frac{1}{4}N(N-1)\hbar\varpi(t),
\end{equation}
and $\varpi(t)$ is defined in Eq.~(\ref{eq:omgbar-def}).

When $\dot{\eta}(t)\neq0$, the inverse square interaction has a different
scaling exponent from the logarithmic interaction. Nevertheless, as
is clear from Eq.~(\ref{eq:Psi-LogCS}), the dynamics given by the
complex-valued TDJA still preserves the scale-invariance. Note that
scale-invariant dynamics only requires that the amplitude of the many-body
wave function scales with $b(t)$, i.e., as in Eq. (\ref{densityinv}).
There are no restrictions on the one-body and two-body phases.

Furthermore, the case of $\eta(t)=\eta_{0}\neq0$ corresponds to the
CS model with interaction strength taking the threshold value $1/4$,
below which the thermodynamic limit does not exist. Nevertheless,
this regime still makes sense as long as the number of particles remains
finite. Our results indicate that in this regime, the eigenstate of
the CS is still of the Jastrow form but with a non-vanishing two-body
phase, which cannot be described by the real-valued TIJA.

Finally, we note that in this case, the prime parent Hamiltonian or
the reference Hamiltonian $\mathscr{\hat{H}}_{0}^{\prime}(t)$ is
the same form as Eq.~(\ref{eq:H0prime-CS}) with $\lambda_{0}$ set
to be $1/2$. The control field for STA now becomes 
\begin{equation}
\mathscr{\hat{H}}_{1}^{\prime}(t)\equiv\frac{\dot{b}(t)}{2b(t)}\sum_{i}\{x_{i},\,p_{i}\}+\frac{\hbar\eta(t)}{4m}\sum_{i<j}\bigg\{\frac{1}{x_{ij}},\,p_{ij}\bigg\}.\label{eq:control-CSB}
\end{equation}
We observe an interesting result: For $\lambda_{0}=1/2$, we actually
give two different counterdiabatic Hamiltonians which can steer the
same adiabatic evolution of $\Phi(t)$, whose control field is given
by Eq.~(\ref{eq:control-CSA}) and Eq.~(\ref{eq:control-CSB}),
respectively. As shown previously, the former satisfies Eq.~(\ref{eq:H1prime-exp})
and the latter is a counterdiabatic driving that does not satisfy
Eq.~(\ref{eq:H1prime-exp}).

The non-vanishing two-body phase angle $\eta(t)$ presents a challenge
in the exact evaluation of the survival probability, even when using
the change of variables in Sec.~\ref{sec:vanishing-eta}. \textcolor{blue}{However,
one can find an upper bound to the survival probability by resorting to the following inequality 
\begin{multline}
|\langle\Psi(t)\big|\Psi(t_{0})\rangle|\leq\frac{1}{[b(t)b(t_{0})|\alpha(t,\,t_{0})|^2]^{N[1+(N-1)\lambda_{0}]/2}}\\
\times\int_{\mathbb{R}^N}d^Nz\prod_{i<j}z_{ij}^{2\lambda_{0}}e^{-\sum_{k}z_{k}^{2}/x_{0}^{2}},
\end{multline}
with $\lambda_{0}=1/2$, which is saturated at the initial time $t_{0}$
due to the cancellation of the two-body phase. Therefore, we obtain
that the survival probability is upper bounded by 
\begin{equation}
\text{Pr}_{s}(t,\,t_{0})\leq \left[\frac{1}{b(t)b(t_{0})|\alpha(t,\,t_{0})|^{2}}\right]^{N[1+\lambda_0(N-1)]}.
\end{equation}
}

\section{\label{sec:The-time-dependent-Hyperbolic} The time-dependent Hyperbolic
models}

Hyperbolic models have been long studied and are characterized by
both hard-core and finite-range interactions. In the limit of low
densities, the finite-range contribution leads to exponentially decaying
interparticle potential reminiscent of the Toda lattice, both in the
homogeneous \citep{Sutherland04} and trapped cases \citep{delcampo20}.
Hyperbolic models also admit generalizations to higher spatial dimensions
\citep{BeauDC21}. Let us focus on their exact quantum dynamics.

We choose the two-body  function $\Gamma_{ij}(t)$ according to the groundstate wavefunction  in the hyperbolic models, 
\begin{align}
\Gamma_{ij}(t) & =\lambda(t)\ln|\sinh[c(t)x_{ij}]|,\\
\Gamma_{ij}^{\prime}(t) & =\lambda(t)c(t)\coth[c(t)x],\\
\Gamma_{ij}^{\prime\prime}(t) & =\textcolor{blue}{-\lambda(t)c^{2}(t)\csch^{2}[c(t)x]},
\end{align}
and 
\begin{equation}
\dot{\Gamma}_{ij}=\lambda(t)\dot{c}(t)\coth[c(t)x_{ij}]x_{ij}+\dot{\lambda}(t)\ln|\sinh[c(t)x]|.
\end{equation}
The two-body consistency condition~(\ref{eq:2-body-consist-main})
leads to the following two cases:

\begin{enumerate}[label={\Alph*)  }]

\item When $ \lambda(t)=\lambda_{0}$, $\eta(t)=0$, $c(t)=c_{0}/b(t)$ (where $c_{0}$ is the initial value of $c(t)$), and $\dot{\widetilde{\mathcal{N}}}_{2}(t)=0$.
Thus, one can take $\mathcal{N}_{23}(t)=0$.

\item When $\lambda_{0}=1/2$, $\eta(t)$ can be arbitrary, and $c(t)$
satisfies 
\begin{equation}
\frac{\dot{c}(t)}{c(t)}=\frac{\eta(t)\omega_{0}}{b^{2}(t)}-\frac{\dot{b}(t)}{b(t)},\label{eq:c-Log-Hyperbolic}
\end{equation}
with 
\begin{equation}
\dot{\widetilde{\mathcal{N}}}_{2}(t)=\frac{\eta(t)c^{2}(t)}{2}, \,\mathcal{\dot{N}}_{23}(t)=\frac{\hbar N(N^{2}-1)\eta(t)c^{2}(t)}{12m}.
\end{equation}
\end{enumerate}

\subsection{The hyperbolic model\label{subsec:The-hyperbolic-model}}

For this case $\Gamma_{ij}(t)$ adopts the form in Eq.~(\ref{eq:vanishing-eta-Gamma}) and 
the complex-valued TDJA~(\ref{eq:Psi-TDJastrow}) becomes 
\begin{equation}
\Psi(t)=\frac{1}{b^{N/2}(t)}\prod_{i<j}\bigg|\sinh\left[\frac{c_{0}x_{ij}}{b(t)}\right]\bigg|^{\lambda_{0}}\prod_{k}e^{-\frac{x_{k}^{2}}{2x_{0}^{2}b^{2}(t)}+\text{i}\frac{\dot{b}(t)x_{k}^{2}}{2\omega_{0}x_{0}^{2}b(t)}+\text{i}\tau(t)}. \label{eq:Psi-Hyperbolic-A}
\end{equation}
The corresponding parent Hamiltonian~(\ref{eq:TDPHJ}) reads 
\begin{eqnarray}
\hat{\mathscr{H}}(t) & =&\frac{1}{2m}\sum_{i}p_{i}^{2}+\frac{1}{2}m\Omega^{2}(t)\sum_{i}x_{i}^{2}\nonumber \\
 & &+\frac{\hbar^{2}}{m}\sum_{i<j}\frac{\lambda_{0}(\lambda_{0}-1)c^{2}(t)}{\sinh^{2}[c(t)x_{ij}]}\nonumber \\
 & &-\hbar\lambda_{0}\omega(t)c(t)\sum_{i<j}\coth[c(t)x_{ij}]x_{ij}+\bar{\mathscr{E}}(t),\label{eq:H-Hyperbolic-A}
\end{eqnarray}
where $\Omega(t)$ is given by the Ermakov equation~(\ref{eq:Ermakov})
and 
\begin{equation}
\bar{\mathscr{E}}(t)=\mathscr{E}(t)+\frac{N(N-1)\hbar^{2}\lambda_{0}^{2}c^{2}(t)}{2m}.
\end{equation}

Since in this case $\eta(t)=0$, the reference Hamiltonian or the
prime parent Hamiltonian is given by Eq.~(\ref{eq:H0prime-CD}),
which has the same functional form as Eq.~(\ref{eq:H-Hyperbolic-A}),
replacing $\Omega(t)$ with $\omega(t)$ and replacing $\bar{\mathscr{E}}(t)$
with $\bar{\mathscr{E}}_{0}^{\prime}(t)$, where 
\begin{equation}
\bar{\mathscr{E}}_{0}^{\prime}(t)=\mathscr{E}_{0}^{\prime}(t)+\frac{N(N-1)\hbar^{2}\lambda_{0}^{2}c^{2}(t)}{2m}.\label{eq:scrE0bar-Hyperbolic-A}
\end{equation}
It is straightforward to that check the auxiliary control term $\hat{\mathscr{H}}_{1}^{\prime}(t)$ is the same as in
Eq.~(\ref{eq:control-CSA}), according to the general expression~(\ref{eq:H1prime-def}).

The hyperbolic model reduces to the CS model for small values of $c_{0}$.
In general, as $c_{0}$ increases the scale-invariance of both $\mathscr{H}_{0}^{\prime}(t)$
and the dynamics given by $\Phi(t)$ or $\Psi(t)$ is violated. Remarkably,
we see that the control field, which keeps the adiabatic evolution
of the complex-valued TDJA does not change when varying $a_{0}$.
It is worth noting that the dynamics given by Eq.~(\ref{eq:Psi-Hyperbolic-A})
is scale-invariant, while for the hyperbolic interaction in Eq.~(\ref{eq:H-Hyperbolic-A}), it
does not have any scaling properties.

In many situations, such as in free expansions, the
scaling factor $b(t)$ satisfies 
\begin{equation}
\lim_{t\to\infty}b(t)=\infty,\,\lim_{t\to\infty}\frac{\dot{b}(t)}{b(t)}=\nu_{\infty}=\text{const}.\label{eq:b-property}
\end{equation}
This is the case in time-of-flight experiments in which particles
are released after suddenly switching off a confining trap. The same asymptotic behavior of $b(t)$ also
applies to quenches  of the interparticle interactions, as we
will see in Sec.~\ref{sec:Quench-dynamics}. For $t\to\infty$, one
finds that $|\alpha(t,\,t_{0})|$ approaches a constant $|\alpha_{\infty}(t_{0})|$
asymptotically, where 
\begin{equation}
|\alpha_{\infty}(t_{0})|\equiv\frac{1}{\sqrt{2}}\left[\frac{1}{b(t_{0})^{4}}+\frac{\text{1}}{\omega_{0}^{2}}\left(\frac{\dot{b}(t_{0})}{b(t_{0})}-\nu_{\infty}\right)^{2}\right]^{1/4}.
\end{equation}

Even though, in this case, the analytic continuation of the integral
in Eq.~(\ref{eq:SA-zvar}) to the real line is not possible, one
still can obtain the long-time asymptotic behavior of the survival
probability given that Eq.~(\ref{eq:b-property}) is satisfied. This is
done by expanding $\sinh\left[c_{0}z_{ij}/\left(\alpha(t,\,t_{0})b(t)\right)\right]$
at late times, leading to \begin{widetext}
\begin{equation}
{\braket{\Psi(t)\big|\Psi(t_{0})}\propto\frac{e^{-\text{i}N[\tau(t)-\tau(t_{0})]}}{[b(t)b(t_{0})]^{N/2}\alpha^{N}(t,\,t_{0})}\int_{\mathscr{C}^N}d^Nz\prod_{i<j}\bigg|\frac{c_{0}z_{ij}}{\alpha_{\infty}(t_{0})b(t)}\sinh \left[\frac{c_{0}z_{ij}}{\alpha_{\infty}(t_{0})b(t_{0})}\right]\bigg|^{\lambda_{0}}e^{-\sum_{k}z_{k}^{2}/x_{0}^{2}}.}
\end{equation}
\end{widetext}Thus, as $t\to\infty$ , asymptotically the survival
probability in the generalized hyperbolic model is given by 
\begin{equation}
\text{Pr}_{s}(t,\,t_{0})\propto b(t)^{-N[1+\lambda_{0}(N-1)]},
\end{equation}
which has the same time scaling as in the CS model.

\subsection{The logarithmic hyperbolic model\label{subsec:The-logarithmic-hyperbolic-model}}

For this case, we  define the dimensionless coefficient 
\begin{equation}
\gamma(t)\equiv\frac{1}{c(t)x_{0}},
\end{equation}
in terms of which
\begin{align}
\eta(t) & =\left[\frac{\dot{b}(t)}{b(t)}-\frac{\dot{\gamma}(t)}{\gamma(t)}\right]\frac{b^{2}(t)}{\omega_{0}},\label{eq:eta-gamma}\\
\mathcal{N}_{23}(t) & =\frac{N(N^{2}-1)b^{2}(t)}{24\gamma^{2}(t)}.
\end{align}
One can also solve $\gamma(t)$ in terms of $\eta(t)$ from Eq.~(\ref{eq:eta-gamma}),
\begin{equation}
\gamma(t)=\gamma_{0}b(t)e^{-\omega_{0}\int_{0}^{t}\frac{\eta(\tau)}{b^{2}(\tau)}d\tau},\label{eq:gamma-exp}
\end{equation}
 where $\gamma_{0}\equiv\gamma(0)$. The complex-valued TDJA~(\ref{eq:Psi-TDJastrow})
becomes 
\begin{eqnarray}
\Psi(\bm{x},\,t) & = & \frac{1}{b^{N/2}(t)e^{N(N-1)b^{2}(t)/[24\gamma^{2}(t)]}}\prod_{i<j}\bigg|\sinh\left[\frac{x_{ij}}{\gamma(t)x_{0}}\right]\bigg|^{[1+\text{i}\eta(t)]/2}\nonumber \\
 &  & \times\prod_{k}e^{-\frac{x_{k}^{2}}{2x_{0}^{2}b^{2}(t)}+\text{i}\frac{\dot{b}(t)x_{k}^{2}}{2\omega_{0}x_{0}^{2}b(t)}+\text{i}\tau(t)}.\label{eq:Psi-Hyperbolic-B}
\end{eqnarray}
The corresponding parent Hamiltonian is 
\begin{align}
\hat{\mathscr{H}}(t) & =\frac{1}{2m}\sum_{i}p_{i}^{2}+\frac{1}{2}m\Omega^{2}(t)\sum_{i}x_{i}^{2}-\frac{\hbar^{2}}{m}\sum_{i<j}\frac{[\eta^{2}(t)+1]c^{2}(t)}{4\sinh^{2}[c(t)x]}\nonumber \\
 & -\frac{1}{2}\hbar\omega(t)c(t)\sum_{i<j}\coth[c(t)x_{ij}]x_{ij}+\bar{\mathscr{E}}(t)\nonumber \\
 & -\frac{1}{2}\hbar\dot{\eta}(t)\sum_{i<j}\ln|\sinh[c(t)x_{ij}]|,\label{eq:H-Hyperbolic-B}
\end{align}
where $\Omega(t)$ is given by the Ermakov equation~(\ref{eq:Ermakov}),
\begin{equation}
\bar{\mathscr{E}}(t)=\mathscr{E}(t)+\frac{N(N-1)\hbar^{2}[1-\eta^{2}(t)]c^{2}(t)}{8m},
\end{equation}
and we have used the identity $\varpi(t)c(t)+\eta(t)\dot{c}(t)=\omega(t)c(t)$
with $\varpi(t)$ defined in Eq.~(\ref{eq:omgbar-def}). We note
that for small value $c(t)$ case A) and B) reduce to the corresponding
cases in the CS and logarithmic CS models.

For non-vanishing $\eta(t)$, as is clear from Eq.~(\ref{eq:eta-gamma}),
$\gamma(t)\neq b(t)$, and the dynamics given by (\ref{eq:Psi-Hyperbolic-B})
is no longer scale-invariant. Finally, we note that in this case the
prime parent Hamiltonian or the reference Hamiltonian $\mathscr{H}_{0}^{\prime}(t)$
is 
\begin{align}
\hat{\mathscr{H}}_{0}^{\prime}(t) & =\frac{1}{2m}\sum_{i}p_{i}^{2}+\frac{1}{2}m\omega^{2}(t)\sum_{i}x_{i}^{2}-\frac{\hbar^{2}}{m}\sum_{i<j}\frac{c^{2}(t)}{4\sinh^{2}[c(t)x_{ij}]}\nonumber \\
 & -\frac{1}{2}\hbar\omega(t)c(t)\sum_{i<j}\coth[c(t)x_{ij}]x_{ij}+\bar{\mathscr{E}}_{0}^{\prime}(t),\label{eq:H0prime-Hyperbolic-B}
\end{align}
where $\bar{\mathscr{E}}_{0}^{\prime}(t)$ is the same as Eq.~(\ref{eq:scrE0bar-Hyperbolic-A})
except now $c(t)$ is expressed by Eq.~(\ref{eq:c-Log-Hyperbolic})
and $\lambda_{0}$ is set to be $1/2$. The auxiliary control field  involves one-and two-body terms
\begin{equation}
\hat{\mathscr{H}}_{1}^{\prime}(t)\equiv\frac{\dot{b}(t)}{2b(t)}\sum_{i}\{x_{i},\,p_{i}\}+\frac{\hbar\eta(t)c(t)}{4m}\sum_{i<j}\left\{ \coth[c(t)x_{ij}],\,p_{ij}\right\}.
\end{equation}

When computing the unnormalized state overlap $|\braket{\Psi(t)\big|\Psi(t_{0})}|$,
as with Eq.~(\ref{eq:SA-zvar}), one can make the  change
of variables $\alpha(t,t_{0})x_{ij}=z_{ij}$. Similarly to the case in  Sec.~\ref{subsec:The-hyperbolic-model},
assuming Eq.~(\ref{eq:b-property}) and $\gamma(t)\to\infty$ as
$t\to\infty$ and performing a Taylor expansion of  $\sinh\left[\frac{|z_{ij}/\alpha(t,t_{0})|}{\gamma(t)x_{0}}\right]$,
one can find the long-time asymptotic upper bound,
\begin{equation}
|\braket{\Psi(t)\big|\Psi(t_{0})}|\lesssim[b(t)]^{-N/2}[\gamma(t)]^{-N(N-1)/4}e^{-N(N-1)b^{2}(t)/[24\gamma^{2}(t)]},
\end{equation}
where $\lesssim$ indicates the time scaling of the l.h.s. is smaller
or equal to that of the r.h.s.

As clear from Eq.~(\ref{eq:gamma-exp}), the behavior $\gamma(t)\to\infty$ as
$t\to\infty$ can be guaranteed if 
\begin{equation}
\int_{0}^{\infty}\frac{\eta(\tau)}{b^{2}(\tau)}d\tau\ll\omega_{0}^{-1}.
\end{equation}
In this case $\gamma(t)\to\gamma_{0}b(t)$ as $t\to\infty$. At late
times, the survival probability of the logarithmic CS model is thus bounded
by
\begin{align}
\text{Pr}_{s}(t,\,t_{0}) & \lesssim[b(t)]^{-N[1+(N-1)/2]}.
\end{align}

\section{\label{sec:TD-long-range-LL}The time-dependent long-range Lieb-Liniger model}

The conventional LL model describes ultracold bosons in one spatial
dimension subject to contact interactions \citep{LL63,L63}, e.g.,
describing $s$-wave scattering in ultracold atoms \citep{Olshanii98}.
The generalization of the LL model with additional one-dimensional
Coulomb or gravitational interactions was introduced in Ref. \citep{BeauPittman20,delcampo20}.
A study of its ground-state properties reveals a rich phase diagram
including a trapped McGuire quantum soliton in the attractive case,
as well as an incompressible Laughlin-like fluid and Wigner-crystal
behavior in the repulsive case. The model has also been generalized
to higher spatial dimensions \citep{BeauDC21} and by changing the
attractive and repulsive character of the interactions \citep{Beau2021dark}.

We next apply our results to discuss its time dependence exactly.
To this end, we take\textit{ 
\begin{align}
\Gamma_{ij}(t)=c(t)|x_{ij}|, & \,\Gamma_{ij}^{\prime}(t)=c(t)\text{{\rm sgn}}(x_{ij}),\\
\Gamma_{ij}^{\prime\prime}(t)=2c(t)\delta_{ij},\, & \dot{\Gamma}_{ij}(t)=\dot{c}(t)|x_{ij}|.
\end{align}
}Then, the only solution consistent with Eq.~(\ref{eq:2-body-consist-main})
involves $\eta(t)=0$, $\dot{\widetilde{\mathcal{N}}}_{2}(t)=0$ and 
\begin{equation}
c(t)=c_{0}\sqrt{\frac{\omega(t)}{\omega_{0}}}=\frac{c_{0}}{b(t)}.
\end{equation}
The function $\Gamma_{ij}(t)$ for case (i) satisfies the property of Eq.~(\ref{eq:vanishing-eta-Gamma}).
Thus, the complex-valued TDJA is 
\begin{equation}
\Psi(\bm{x},\,t)=\frac{1}{b^{N/2}(t)}\prod_{i<j}e^{c_{0}|x_{ij}/b(t)|}e^{-\frac{\sum_{k}x_{k}^{2}}{2x_{0}^{2}b^{2}(t)}+\text{i}\frac{\dot{b}(t)\sum_{k}x_{k}^{2}}{2\omega_{0}x_{0}^{2}b(t)}+\text{i}N\tau(t)}.\label{eq:Psi-LL}
\end{equation}
Note that as long as $\omega(t)$ is strictly positive, the sign of
$c_{0}$ can be arbitrary. In the case where $\omega(t)=0$, $c(t)$
must be negative so that $\Psi$ is still normalizable. According
to Eq.~(\ref{eq:TDPHJ}), the corresponding parent Hamiltonian reduces
to 
\begin{align}
\hat{\mathscr{H}}(t) & =\frac{1}{2m}\sum_{i}p_{i}^{2}+\frac{1}{2}m\Omega^{2}(t)\sum_{i}x_{i}^{2}+2c(t)\sum_{i<j}\delta_{ij}\nonumber \\
 & -\hbar\omega(t)c(t)\sum_{i<j}|x_{ij}|+\bar{\mathscr{E}}(t),\label{eq:PHJ-TDCV-LL}
\end{align}
where $\Omega(t)$ is given by the Ermakov equation~(\ref{eq:Ermakov})
and 
\begin{equation}
\bar{\mathscr{E}}(t)=\mathscr{E}(t)+\frac{N(N-1)\hbar^{2}c^{2}(t)}{2m}.
\end{equation}

A few comments are in order. The time-independent
version of Eq.~(\ref{eq:PHJ-TDCV-LL}) was found previously~\citep{delcampo20}
in the context of the real-valued TIJA, and is the long-range LL model,
where the long-range interaction is the Coulomb repulsion in one dimension
\citep{BeauPittman20}. Further, it is worth noting that in the limit
$c_{0}\to\infty$ the long-range interaction term does not vanish
and it does not reduce to the scale-invariant dynamics for Tonks-Girardeau
gas~\citep{minguzzi2005exactcoherent}. Instead, it can be recovered
by taking limit $\lambda_{0}\to1$ in the Calogero model discussed
above (see the Supplemental Material of Ref.~\citep{yang2022one}).
Finally, the dynamics is, in this case, still scale-invariant as one
can see from Eq.~(\ref{eq:PHJ-TDCV-LL}), but the interaction in
the parent Hamiltonian no longer has the scaling property.

When applying to STA, the reference Hamiltonian has the same form
as Eq.~(\ref{eq:PHJ-TDCV-LL}) with $\Omega(t)$ replaced by $\omega(t)$
and $\bar{\mathscr{E}}(t)$ replaced by $\bar{\mathscr{E}}_{0}^{\prime}(t)$,
where 
\begin{equation}
\bar{\mathscr{E}}_{0}^{\prime}(t)=\mathscr{E}_{0}^{\prime}(t)+\frac{N(N-1)\hbar^{2}c^{2}(t)}{2m}.
\end{equation}
The auxiliary control field is given by Eq.~(\ref{eq:control-CSA}) and is
thus the same as the one for the CS model.

For a scaling factor $b(t)$ that satisfies Eq.~(\ref{eq:b-property}),
it is straightforward to see that as $t\to\infty$, to the leading
order,
\begin{equation}
\braket{\Psi(t)\big|\Psi(t_{0})}\sim\frac{\int_{\mathscr{C}^N}d^Nze^{\frac{c_{0}}{b(t_{0})}\big|\frac{z_{ij}}{\alpha_{\infty}(t_{0})}\big|-\sum_{k}z_{k}^{2}/x_{0}^{2}}}{[b(t)b(t_{0})]^{N/2}|\alpha_{\infty}^{N}(t_{0})|}.
\end{equation}
We conclude that for a fixed number of particles $N$, the asymptotic
time-dependence of $|\braket{\Psi(t)\big|\Psi(t_{0})}|$ is given
by $1/[b(t)]^{N/2}$ so that 
\begin{equation}
\text{Pr}_{s}(t,\,t_{0})\propto\frac{1}{[b(t)]^{N}},\,\text{as }t\to\infty.\label{eq:SP-LL}
\end{equation}
We shall check this prediction against numerical calculations in the next
section for specific choices of $b(t)$.

\textcolor{blue}{This result describes the asymptotic decay of an initially trapped bright many-body quantum soliton in a one-dimensional expansion. It is noteworthy that the long-time behavior differs from the case when the initial state is subject to hard-core repulsive interactions when the decay is set by $\text{Pr}_{s}(t,\,t_{0})\propto[b(t)]^{-N^2}$ \cite{delcampo11,delcampo16}.}

\section{\label{sec:Effective-dynamics}Effective dynamics and the one-body
reduced density matrix (OBRDM)}

As we have discussed in Sec.~\ref{sec:Imposing-Hermicitiy}, the
Ermakov equation~(\ref{eq:Ermakov}) is a second-order nonlinear
differential equation and we only impose $b(0)=1$, while leaving $\dot{b}(0)$
unspecified. If we further impose the initial condition 
\begin{equation}
\dot{b}(0)=0,\,\eta(0)=0,\,\dot{\eta}(0)=0,\label{eq:additional-IC}
\end{equation}
we find according to Eqs.~(\ref{eq:scr-U-harmonic}) and (\ref{eq:H1prime-def}) that
\begin{equation}
\mathscr{U}(0)=e^{{\rm i}N\tau(0)},\,\Psi(0)=\Phi(0)e^{{\rm i}N\tau(0)},\,\mathscr{H}_{1}^{\prime}(0)=0.
\end{equation}
If additionally one chooses
\begin{equation}
\Omega_{0}=\omega_{0},\label{eq:same-initial-frequency}
\end{equation}
so that $\Omega(t)$ and $\omega(t)$ coincide at $t=0$, it follows from
 Eqs.~(\ref{eq:TDPHJ}) and (\ref{eq:H0prime-CD}) that
\begin{equation}
\hat{\mathscr{H}}(0)=\mathscr{\hat{H}}_{0}^{\prime}(0).
\end{equation}
Since $\hat{\mathscr{H}}_{0}^{\prime}(0)\Phi(0)=0$, Eqs.~(\ref{eq:additional-IC}) and (\ref{eq:same-initial-frequency})
would imply $\hat{\mathscr{H}}(0)\Psi(0)=0$. That is, the complex-valued
TDJA not only satisfies the time-dependent Schr\"odinger equation, but
its projection at time $t=0$ is also an eigenstate of $\hat{\mathscr{H}}(0)$.

However, we should keep in mind that the conditions specified in Eqs.~(\ref{eq:additional-IC}) 
and (\ref{eq:same-initial-frequency}) correspond only to one of the possible
choices. In fact, by changing these conditions, one can account for
the quench dynamics resulting from fast changes of $b(t)$ or $\eta(t)$
as we will discuss in the next section.

\textcolor{blue}{Furthermore, for scale-invariant dynamics that appears
in the previous sections, 
given (\ref{eq:additional-IC}), the TDJA can be universally represented
as follows: 
\begin{equation}
\Psi(\bm{x},\,t)=\mathscr{U}_{P}(\bm{x},\,t)\mathscr{U}_{S}(\bm{x},\,t)\Psi(\bm{x},\,t=0),\label{eq:Scale-invariant-Universal}
\end{equation}
where the scale transformation is defined as 
\begin{equation}
\mathscr{U}_{S}(\bm{x},\,t)\equiv e^{-\text{i}\ln b(t)\sum_{k}[x_{k}p_{k}+p_{k}x_{k}]/(2\hbar)}.
\end{equation}
A few comments are in order. (i) Without assuming Eq. ~(\ref{eq:additional-IC}),
a similar representation is also possible, but $\mathscr{U}_{P}(\bm{x},\,t)$
and $\mathscr{U}_{S}(\bm{x},\,t)$ need to be adapted to
account for the non-vanishing value of $\dot{b}(t)$ and $\eta(t)$
at the initial time. (ii) The product of unitaries $\mathscr{U}_{P}(\bm{x},\,t)\mathscr{U}_{S}(\bm{x},\,t)$
may generically differ from the full unitary evolution $\mathcal{T}\exp\left[\int_{0}^{t}\hat{\mathscr{H}}(\tau)d\tau\right]$,
where $\mathcal{T}$ denotes time ordering. In other words, $\mathscr{U}_{P}(\bm{x},\,t)\mathscr{U}_{S}(\bm{x},\,t)$
may be viewed as describing an }\textit{\textcolor{blue}{effective}}\textcolor{blue}{{}
quantum many-body dynamics when the full unitary evolution is applied
on the TDJA at time $t=0$. While for the CS model discussed in
Sec.~\ref{subsec:CS-Scale-Invariant}, they are the same~\citep{delcampo2013shortcuts,gritsev2010scaling},
their relations in other models are left for future studies. }

\textcolor{blue}{
Consider next the study of local observables, which can be computed when the 
the one-body reduced density matrix
(OBRDM) is known. 
The OBRDM is defined as 
\begin{equation}
\rho_{1}(x,\,x^{\prime},\,t)=\frac{N}{\mathcal{C}}\int_{\mathbb{R}^N}\Psi(x,\,\bar{\bm{x}},\,t)\Psi^{*}(x^{\prime},\,\bar{\bm{x}},\,t)d^{N-1}\bar{x},
\end{equation}
where $\bar{\bm{x}}\equiv(x_{2},\,\cdots,\,x_{N})$, $d^{N-1}\bar{x}=\prod_{m=2}^Ndx_m$, and $\mathcal{C}$
is defined in Eq.~(\ref{eq:calC-def}). 
To shed light on its evolution, one can use Eq.~(\ref{eq:Scale-invariant-Universal}).
We note that $\mathscr{U}_{S}(\bm{x},\,t)$
only involves one-body operation while this is only true for  $\mathscr{U}_{P}(\bm{x},\,t)$
 when $\eta(t)=0$. As  discussed in Sec.~\ref{sec:vanishing-eta},
$\Gamma_{ij}(t)=\Gamma(c_{0}x_{ij}/b(t))$, }so that \begin{widetext}\textcolor{blue}{
\begin{equation}
\rho_{1}(x,\,x^{\prime},\,t)=\frac{N}{\mathcal{C}b^{N}(t)}e^{-\frac{x^{2}+x^{\prime2}}{2x_{0}^{2}b^{2}(t)}+\text{i}\frac{(x^{2}-x^{\prime2})\dot{b}(t)}{2x_{0}^{2}\omega_{0}b(t)}}\int_{\mathbb{R}^N}d^{N-1}\bar{x}e^{2\sum_{i<j,\,i,j\ge2}\Gamma\left(\frac{c_{0}x_{ij}}{b(t)}\right)+\sum_{j\ge2}\left[\Gamma\left(\frac{c_{0}(x-x_{j})}{b(t)}\right)+\Gamma\left(\frac{c_{0}(x^{\prime}-x_{j})}{b(t)}\right)\right]-\frac{\sum_{k\ge2}x_{k}^{2}}{x_{0}^{2}b^{2}(t)}}.\label{eq:rho1-explicit}
\end{equation}
Making change of variables $x/b(t)\to y$, $x^{\prime}/b(t)\to y^{\prime}$and
$x_{k}/b(t)\to y_{k}$ for fixed time $t$, we find 
\begin{equation}
\rho_{1}\left(y(t),\,y^{\prime}(t),\,t\right)=\frac{N}{\mathcal{C}b(t)}e^{-\frac{y^{2}(t)+y^{\prime2}(t)}{2x_{0}^{2}}+\text{i}\frac{[y^{2}(t)-y^{\prime2}(t)]\dot{b}(t)b(t)}{2x_{0}^{2}\omega_{0}}}\int_{\mathbb{R}^N}d^{N-1}\bar{y}e^{2\sum_{i<j,\,i,j\ge2}\Gamma\left(c_{0}y_{ij}\right)+\sum_{j\ge2}\left[\Gamma\left(c_{0}(y(t)-y_{j})\right)+\Gamma\left(c_{0}(y^{\prime}(t)-y_{j})\right)\right]-\frac{\sum_{k\ge2}y_{k}^{2}}{x_{0}^{2}}}.\label{eq:rho1-explicit-1}
\end{equation}
The numerical calculation of the OBRDM and the density file $\rho(x,\,t)\equiv\rho_{1}(x,\,x^{\prime},\,t)$
only involves real multi-dimensional integrals. On the other hand,
$\mathcal{C}$, the normalization of the TDJA, can be efficiently
computed numerically as we have argued in Sec.~\ref{sec:vanishing-eta}.
Thus, both the OBRDM and the density profile can be computed efficiently
with the Monte Carlo algorithms. The momentum distribution defined
as $n(p,\,t)\equiv\frac{1}{2\pi\hbar}\int dxdx^{\prime}e^{-\text{i}p(x-x^{\prime})}\rho_{1}(x,\,x^{\prime},\,t)$
can be also calculated with Eq.~(\ref{eq:rho1-explicit}). }

\end{widetext}

 \textcolor{blue}{Moreover, one can show analytically the OBRDM also
follows the unitary dynamics in this case since with Eq.~(\ref{eq:Scale-invariant-Universal})
the OBRDM can be easily calculated: 
\begin{align}
\rho_{1}(x,\,x^{\prime},t) & =\mathscr{U}_{P}(x,\,t)\mathscr{U}_{S}(x,\,t)\rho_{1}(x,\,x^{\prime},t=0)\mathscr{U}_{S}^{\dagger}(x^{\prime},\,t)\mathscr{U}_{P}^{\dagger}(x^{\prime},\,t),\nonumber \\
 & =\frac{1}{b(t)}\exp\left[\text{i}\frac{\dot{b}(t)(x^{2}-x^{\prime2})}{2\omega_{0}x_{0}^{2}b(t)}\right]\rho_{1}\left(\frac{x}{b(t)},\,\frac{x^{\prime}}{b(t)},t=0\right).\label{eq:OBRDM-dyn}
\end{align}
}

\textcolor{blue}{This relation has been reported in Refs.~\citep{minguzzi2005exactcoherent,Dupays22}
for the scale-invariant dynamics of Tonks-Girardeau gases in the context
of dynamical fermionization and its control. Starting from Eq.~(\ref{eq:OBRDM-dyn})
and based on an analysis of the Wigner function, Ref.~\citep{Dupays22}
found the relations between the initial density distributions and
the later time density distribution, as well as the momentum distribution. }

\textcolor{blue}{Here, we extend Eq.~(\ref{eq:OBRDM-dyn}) to a variety
of models, including the family of CS models, the hyperbolic models
and the long-range LL model, where the exact dynamics given by the
TDJA is scale-invariant with $\eta(t)=0$. Therefore, as in Refs.~\citep{minguzzi2005exactcoherent,Dupays22},
one can conclude that}

\textcolor{blue}{
\begin{equation}
\rho(x,t)=\frac{1}{b(t)}\rho\left(\frac{x}{b(t)},\,0\right),
\end{equation}
holds for all times, where $\rho(x,\,t)\equiv\rho_{1}(x,x,t)$. At
late times where $t\to\infty$, when $b(t)$ satisfies Eq.~(\ref{eq:b-property}),
one can show that 
\begin{equation}
n(p,t)\approx\frac{1}{m\dot{b}(t)}\rho\left(\frac{p}{m\dot{b}(t)},\,0\right).
\end{equation}
}

\textcolor{blue}{We conclude this section by noting that for the logarithmic
CS model, even if the dynamics is scale-invariant, $\mathscr{U}_{P}(\bm{x},\,t)$
involves two-body operations, and thus, the time evolution of the OBRDM in this model is no
longer unitary.}

\section{\label{sec:Quench-dynamics}Quench dynamics}

Although in previous sections, we have investigated the parent Hamiltonians
corresponding to the TDJA, it is worth discussing the quench limit
where some parameter(s) in the Hamiltonian change from one value to
another suddenly. For the sake of illustration, we consider interaction quenches in the logarithmic CS model and the LL model by specifying the functional
forms of $b(t)$ or $\eta(t)$. A similar analysis is also applicable
to the hyperbolic model and the logarithmic hyperbolic model.

\subsection{Interaction quench of the logarithmic Calogero-Sutherland gases}

The TDJA leads to the scaling-invariant dynamics for the standard
CS model,  discussed in Sec.~\ref{sec:TD-CS}, a feature
extensively used in the previous literature~\citep{Sutherland98,DaeYupSong01,DaeYupSong02,delcampo16}.
Quenches in this model of the trap frequency have been discussed in
Ref.~\citep{Sutherland98,delcampo16}. Here, we focus on the time-dependent
logarithmic CS model~(\ref{eq:TDPHJ-CS}), where the interaction
between the particles is quenched from one value to another.
We assume $\Omega(t)=\Omega_{0}=\omega_{0}$, implying $b(t)=1$.
We take 
\begin{equation}
\eta(t)=\eta_{0}\Theta(-t),
\end{equation}
and Eq.~(\ref{eq:TDPHJ-CS}) becomes 
\begin{align}
\mathscr{H}(t) & =\frac{1}{2m}\sum_{i}p_{i}^{2}+\frac{1}{2}m\Omega_{0}^{2}\sum_{i}x_{i}^{2}-\frac{\hbar^{2}}{4m}\sum_{i<j}\frac{1}{x_{ij}^{2}}\nonumber \\
 & +\begin{cases}
-\frac{\hbar^{2}}{4m}\sum_{i<j}\frac{\eta_{0}^{2}}{x_{ij}^{2}} & t=0^{-}\\
\delta(t)\frac{\hbar\eta_{0}}{2}\sum_{i<j}\ln|x_{ij}| & t=0\\
0 & t=0^{+}
\end{cases},
\end{align}
where we note that sudden quench in $\eta(t)$ is now supplemented
by the delta-kick at $t=0$. Under such a quench, the complex-valued
TDJA ~(\ref{eq:Psi-LogCS}) reads 
\begin{equation}
\Psi(t)=\prod_{i<j}|x_{ij}|^{[1+\text{i}\eta_{0}\Theta(-t)]/2}e^{-\frac{m\Omega_{0}}{2\hbar}\sum_{k}x_{k}^{2}+\text{i}N\tau(t)},
\end{equation}
where $\tau(t)$ is determined by imposing $\bar{\mathscr{E}}(t)=0$
and is found to be $\tau(t)=-\frac{1}{4}(N+1)\Omega_{0}t$. We see
that there is a phase slip in the two-body wave function.

\begin{table*}
\begin{centering}
\begin{tabular}{c||cc||c}
\hline 
 & \multirow{1}{*}{Dynamics} & \multirow{1}{*}{Parent Hamiltonian} & \multirow{1}{*}{Constraints}\tabularnewline
 &  &  & \tabularnewline
\hline 
\hline 
\multirow{2}{*}{Generic model} & \multirow{2}{*}{$\Psi(t)=\frac{1}{b^{N/2}(t)e^{\mathcal{N}_{23}(t)}}\prod_{i<j}e^{\Gamma_{ij}(t)[1+\text{i}\eta(t)]}\prod_{k}e^{-\frac{x_{k}^{2}}{2x_{0}^{2}b^{2}(t)}+\text{i}\frac{\dot{b}(t)x_{k}^{2}}{2\omega x_{0}^{2}b(t)}+\text{i}\tau(t)}$} & \multirow{2}{*}{Eq.~(\ref{eq:TDPHJ})} & \multirow{2}{*}{Eq.~(\ref{eq:2-body-consist-main})}\tabularnewline
 &  &  & \tabularnewline
\hline 
\multirow{2}{*}{STA} & \multirow{2}{*}{$\Phi(t)=\frac{1}{b^{N/2}(t)e^{\mathcal{N}_{23}(t)}}\prod_{i<j}e^{\Gamma_{ij}(t)}\prod_{k}e^{-\frac{x_{k}^{2}}{2x_{0}^{2}b^{2}(t)}}$} & \multirow{2}{*}{Eq.~(\ref{eq:Hprime-CD})} & Eq.~(\ref{eq:2-body-consist-main})\tabularnewline
 &  &  & $\mathscr{H}_{0}^{\prime}(t)\Phi=0$\tabularnewline
\hline 
\multirow{2}{*}{$\eta(t)=0$} & \multirow{2}{*}{$\Gamma_{ij}(t)=\Gamma\left(\frac{c_{0}x_{ij}}{b(t)}\right)$, $\Gamma$
is a generic even function} & \multirow{2}{*}{Eq.~(\ref{eq:H-vanishing-eta})} & \tabularnewline
 &  &  & \tabularnewline
\cline{1-3} \cline{2-3} \cline{3-3} 
\multirow{2}{*}{Time-dependent scale-invariant CS} & \multirow{2}{*}{$\Gamma_{ij}(t)=\lambda_{0}\ln|x_{ij}|$, $\eta(t)=0$} & \multirow{2}{*}{Eq.~(\ref{eq:TDPHJ-CS-regular})} & \tabularnewline
 &  &  & \tabularnewline
\cline{1-3} \cline{2-3} \cline{3-3} 
\multirow{2}{*}{Logarithmic CS} & \multirow{2}{*}{$\Gamma_{ij}(t)=\ln|x_{ij}|/2$, arbitrary $\eta(t)$} & \multirow{2}{*}{Eq.~(\ref{eq:TDPHJ-CS})} & Eq.~(\ref{eq:2-body-consist-main})\tabularnewline
 &  &  & \tabularnewline
\cline{1-3} \cline{2-3} \cline{3-3} 
\multirow{2}{*}{Time-dependent hyperbolic} & \multirow{2}{*}{$\Gamma_{ij}(t)=\lambda_{0}\ln|\sinh[c(t)x_{ij}]|$, $\eta(t)=0$} & \multirow{2}{*}{Eq.~(\ref{eq:H-Hyperbolic-A})} & is satisfied\tabularnewline
 &  &  & \tabularnewline
\cline{1-3} \cline{2-3} \cline{3-3} 
\multirow{2}{*}{Non scale-invariant logarithmic hyperbolic} & \multirow{2}{*}{$\Gamma_{ij}(t)=\ln|\sinh[c(t)x_{ij}]|/2$, arbitrary $\eta(t)$} & \multirow{2}{*}{Eq.~(\ref{eq:H-Hyperbolic-B})} & \tabularnewline
 &  &  & \tabularnewline
\cline{1-3} \cline{2-3} \cline{3-3} 
\multirow{2}{*}{Long-range LL} & \multirow{2}{*}{$\Gamma_{ij}(t)=\exp[c(t)|x_{ij}|]$, $\eta(t)=0$} & \multirow{2}{*}{Eq.~(\ref{eq:PHJ-TDCV-LL})} & \tabularnewline
 &  &  & \tabularnewline
\hline 
\end{tabular}
\par\end{centering}
\caption{\label{tab:Summary-of-finding}The family of parent Hamiltonians of
a time-dependent Jastrow ansatz, which includes evolutions characterized
by scale invariance as well as instances lacking it, such as the logarithmic
hyperbolic model. At variance with previous results discussed in the
literature \citep{gritsev2010scaling,delcampo2013shortcuts,deffner2014classical},
the class of scale-invariant evolutions is not restricted to Hamiltonians
with scale-invariant interactions, such as the scale-invariant time-dependent
CS model, but accommodates for cases where interactions do not have
the same scaling dimension as the kinetic energy operator.}
\end{table*}


\subsection{Interaction quench of the long-range Lieb-Liniger gases}

As a second example, we consider the quench of interactions in the
LL model. We first make change of variables $\beta(t)=1/b(t)$, Eq.~(\ref{eq:PHJ-TDCV-LL})
becomes 
\begin{align}
\mathscr{H}(t) & =\frac{1}{2m}\sum_{i}p_{i}^{2}+\frac{1}{2}m\Omega^{2}(t)\sum_{i}x_{i}^{2}+2c_{0}\beta(t)\sum_{i<j}\delta_{ij}\nonumber \\
 & -\hbar\omega_{0}c_{0}\beta^{3}(t)\sum_{i<j}|x_{ij}|+\bar{\mathscr{E}}(t),
\end{align}
and the Ermakov equation~(\ref{eq:Ermakov}) becomes 
\begin{equation}
\Omega^{2}(t)=\omega_{0}^{2}\beta^{4}(t)-\frac{2\dot{\beta}^{2}(t)}{\beta^{2}(t)}+\frac{\ddot{\beta}(t)}{\beta(t)}.
\end{equation}
Consider the case in which  the interaction is varied by setting 
\begin{equation}
\beta(t)=\frac{{2}}{1+e^{\kappa t}},
\end{equation}
where $\text{\ensuremath{\kappa}}$ is sufficiently large so that
$\beta(t)$ can be approximated by $\Theta(-t)$ and the factor of
$2$ is to satisfy the condition $\beta(0)=b(0)=1$.

Instead of considering the initial time $t_{0}=0$, we shall take
the initial time $t_{0}<0$ with $|t_{0}|\ll1/\kappa$. Next, imposing
$\bar{\mathscr{E}}(t)=0$, one can find the expression for $\tau(t)$.
In particular, for $t\gg\kappa^{-1}$, we find 
\begin{equation}
\tau(t)\approx\tau_{\infty}\equiv\left[2\omega_{0}-\frac{2\hbar(N^{2}-1)c_{0}^{2}}{3m}\right]t_{0},
\end{equation}
where we have used the approximation $\int_{t_{0}}^{t}\frac{1}{(1+e^{\kappa s})^{2}}ds\approx-t_{0}$
for $t\gg\kappa^{-1}$. 
\begin{figure}[t]
\centering \includegraphics[width=0.9\linewidth]{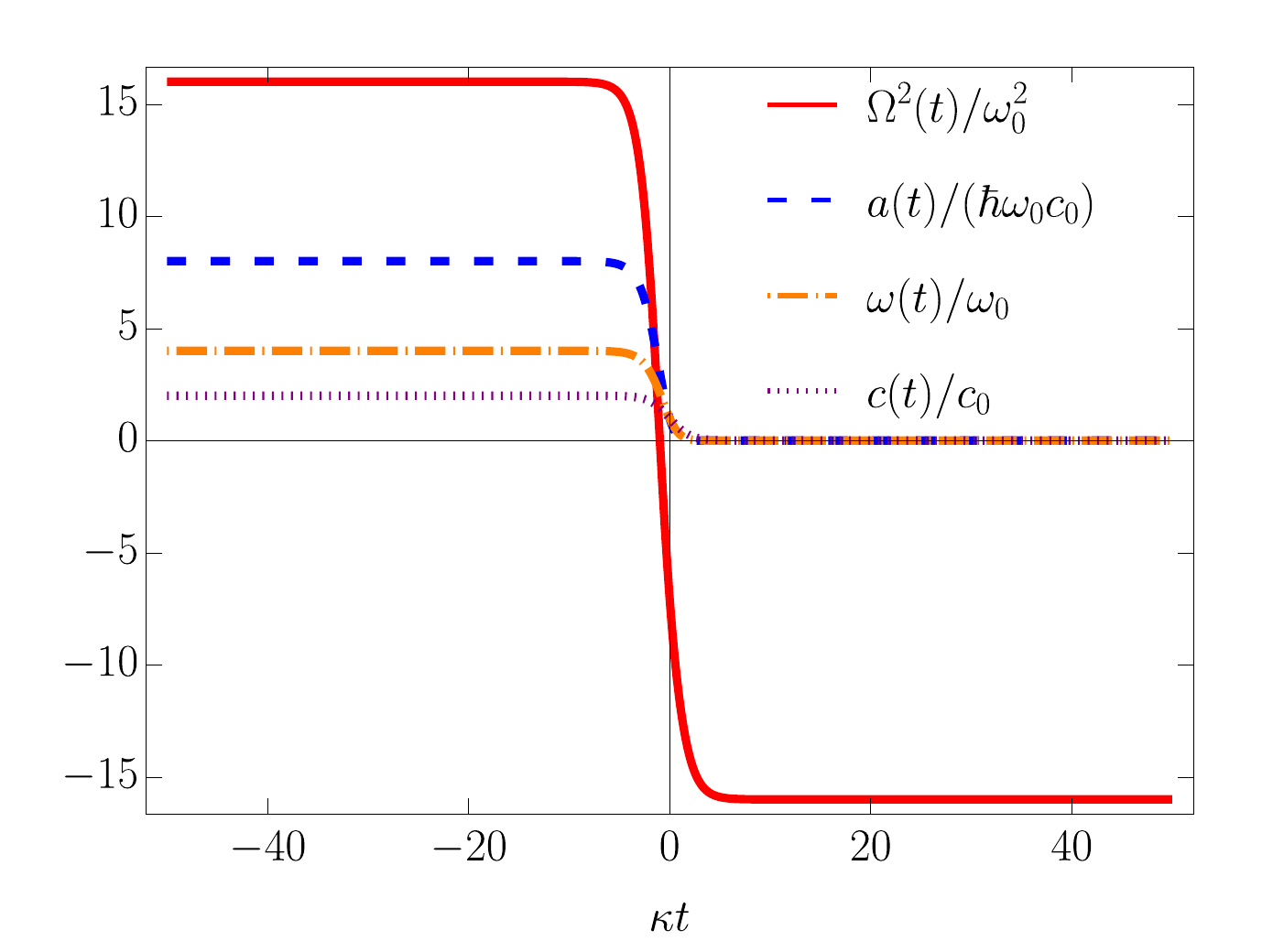} \centering
\caption{Engineered time-dependence of the control parameters as a function
of $\kappa t$. For the plot of $\Omega^{2}(t)/\omega_{0}^{2}$, $\kappa/\omega_{0}=4$.
The initial time is $t_{0}=-50\kappa^{-1}$ and $a(t)$ is the strength
of the long-range Coulomb interactions defined as $a(t)\equiv\hbar\omega_{0}c_{0}\beta^{3}(t)$.
For $\kappa t\gg1$, the time dependence of all these parameters is
well approximated by the Heaviside function and, therefore, leads to
a sudden-quench protocol.}
\label{fig:CtrlPara}
\end{figure}

Therefore, according to Fig.~\ref{fig:CtrlPara} for $\kappa t\gg1$,
we find 
\begin{align}
\mathscr{H}(t) & =\frac{1}{2m}\sum_{i}p_{i}^{2}+\bar{\mathscr{E}}(t)\nonumber \\
 & +\begin{cases}
+8m\omega_{0}^{2}\sum_{i}x_{i}^{2}+4c_{0}\sum_{i<j}\delta_{ij} & t=t_{0}\\
\qquad\quad-8\hbar\omega_{0}c_{0}\sum_{i<j}|x_{ij}|\\
-\frac{1}{2}m\kappa^{2}\sum_{i}x_{i}^{2} & t\gg1/\kappa
\end{cases},
\end{align}
where we have used the approximation that $\kappa$ is large enough
such that $\kappa t\gg1$. The complex-valued TDJA is given by Eq.~(\ref{eq:Psi-LL})
with $b(t)=(1+e^{\kappa t})/2$. In particular, we note that 
\begin{equation}
\Psi(t)=\begin{cases}
2^{N/2}\prod_{i<j}e^{2c_{0}|x_{ij}|}e^{-\frac{2m\omega_{0}}{\hbar}\sum_{k}x_{k}^{2}} & t=t_{0}\\
2^{N/2}e^{-N\kappa t/2}\prod_{i<j}e^{2c_{0}|x_{ij}|/e^{\kappa t}} & t\gg1/\kappa\\
\qquad\quad\times e^{-\frac{2m\omega_{0}}{\hbar}\sum_{k}x_{k}^{2}/e^{2\kappa t}+\text{i}\frac{m\kappa}{2\hbar}\sum_{k}x_{k}^{2}+\text{i}N\tau_{\infty}}
\end{cases},
\end{equation}
for  $\kappa t\gg1$. In particular, we note that
the case of $\omega_{0}=0$ with $c_{0}<0$ corresponds to the dynamics
of quenching the attractive interactions in the LL model from the
McGuire soliton state and then applying an inverted harmonic trap.
\begin{figure}
\begin{centering}
\includegraphics[width=0.95\linewidth]{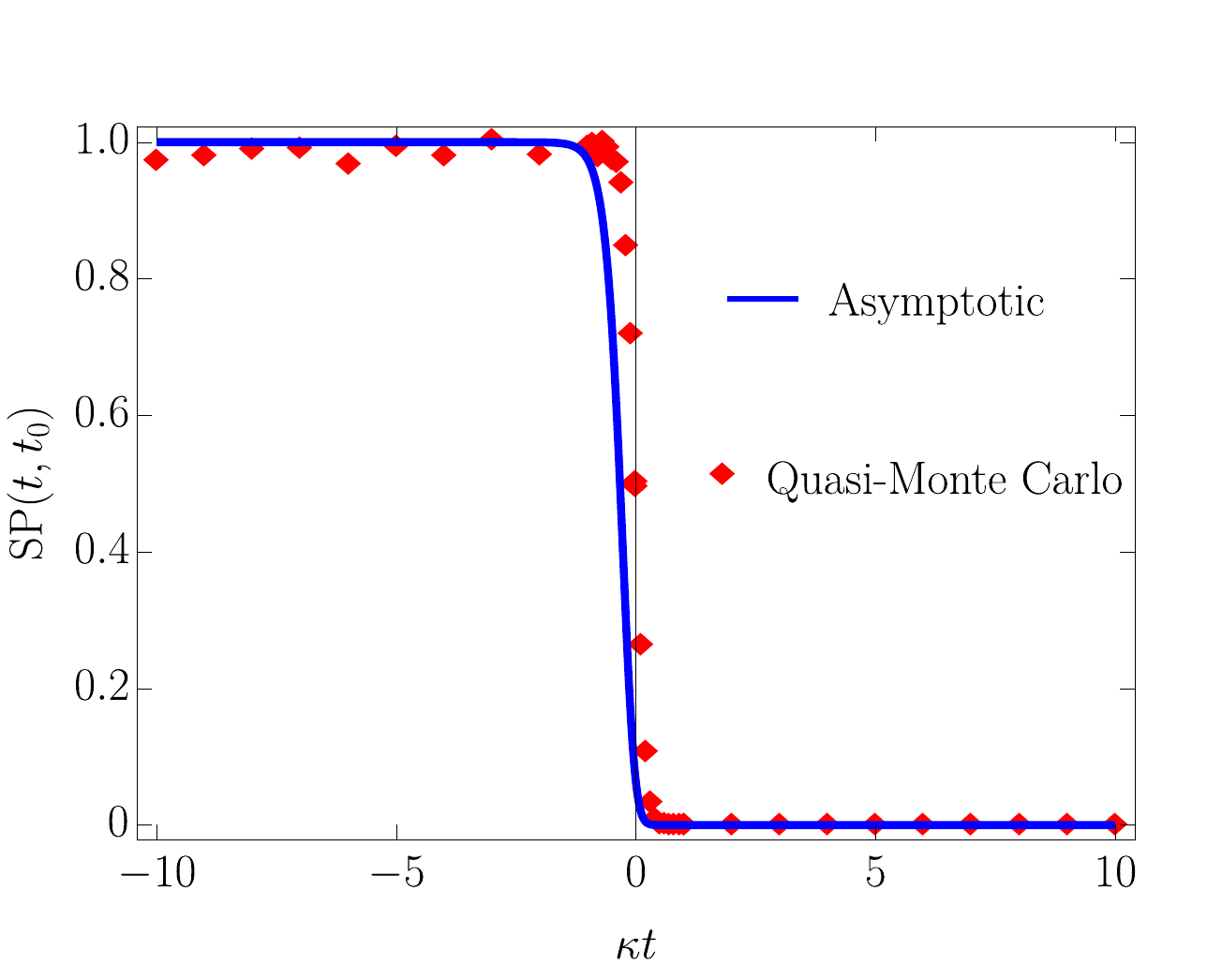}
\par\end{centering}
\caption{\label{fig:SP-LL-InteractionQuench}The numerical calculation versus
the analytic prediction of the survival probability for $4$ particles
for the quench of the interactions in the long-range LL model. The
initial time $t_{0}=-50\kappa^{-1}$. The values of parameters: $\hbar=1$,
$N=4$, $\omega_{0}=1$, $\kappa=5$ and $c_{0}=-1$. For this case,
the survival probability is perfectly predicted by Eq.~(\ref{eq:SP-LL}).}
\end{figure}

\textcolor{blue}{To validate our asymptotic predictions of the survival probabilities, we perform numerical simulations in Mathematica, using the quasi-Monte Carlo algorithm, which utilizes quasi-random numbers \cite{QuasiMonteCarloMathWorld}}. The quasi-Monte Carlo calculation of the survival probability is displayed in Fig.~\ref{fig:SP-LL-InteractionQuench}, in close agreement with the asymptotic
predication of the long-time decay given by Eq.~(\ref{eq:SP-LL}).

\subsection{Sudden release of the long-range Lieb-Liniger gases}

As the last example for the time-dependent long-range LL model~(\ref{eq:PHJ-TDCV-LL}),
we assume the initial time $t_{0}=0$ and impose Eqs.~(\ref{eq:additional-IC}) and (\ref{eq:same-initial-frequency}).
We consider a sudden quench of the trapping frequency familiar in time-of-flight measurements from an  initial frequency $\Omega_{0}$
to a final zero value. In this case, the scaling factor obeys $b(t)=\sqrt{1+\omega_{0}^{2}t^{2}}$.
The time-dependence of all the relevant parameters is shown in Fig.~\ref{fig:CtrlPara-QuenchFreq}.
\begin{figure}[t]
\centering \includegraphics[width=0.85\linewidth]{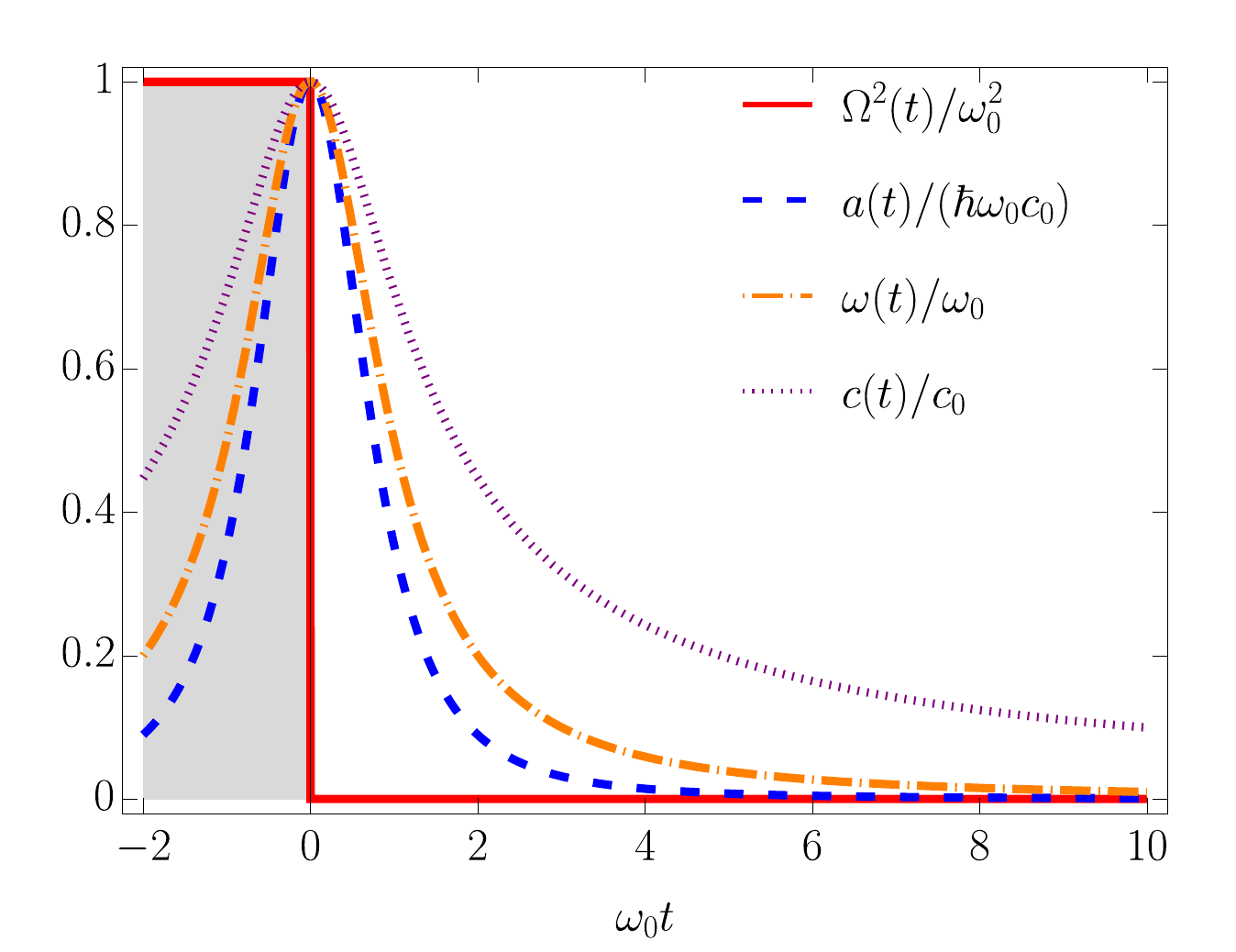} \centering
\caption{Engineered time-dependence of the control parameters as a function
of $\omega_{0}t$. The initial time $t_{0}=0$ and $a(t)$ is the
strength of the long-range Coulomb interactions defined as $a(t)\equiv\hbar\omega_{0}c_{0}/b^{3}(t)$.}
\label{fig:CtrlPara-QuenchFreq}
\end{figure}

Imposing $\bar{\mathscr{E}}(t)=0$, we find 
\begin{align}
\tau(t) & =\left[\frac{\hbar(N^{2}-1)c_{0}^{2}}{6m}-\frac{1}{2}\omega_{0}\right]\int_{0}^{t}\frac{ds}{b^{2}(s)}\nonumber \\
 & =\left[\frac{\hbar(N^{2}-1)c_{0}^{2}}{6m\omega_{0}}-\frac{1}{2}\right]\arctan(\omega_{0}t).
\end{align}
The Hamiltonian ~(\ref{eq:PHJ-TDCV-LL}) becomes 
\begin{align}
\mathscr{H}(t) & =\frac{1}{2m}\sum_{i}p_{i}^{2}+\frac{2c_{0}}{b(t)}\sum_{i<j}\delta_{ij}-\frac{\hbar\omega_{0}c_{0}}{b^{3}(t)}\sum_{i<j}|x_{ij}|\nonumber \\
 & +\begin{cases}
\frac{1}{2}m\omega_{0}^{2}\sum_{i}x_{i}^{2} & t=0^{-}\\
0 & t\ge0
\end{cases},
\end{align}
where TDJA is given by Eq.~(\ref{eq:Psi-LL}). The quasi-Monte Carlo
calculation of survival probability is displayed in Fig.~\ref{fig:SP-LL-TrapQuench},
validating the asymptotics for the long-time decay given by Eq.~(\ref{eq:SP-LL}).

\begin{figure}
\begin{centering}
\includegraphics[width=0.9\linewidth]{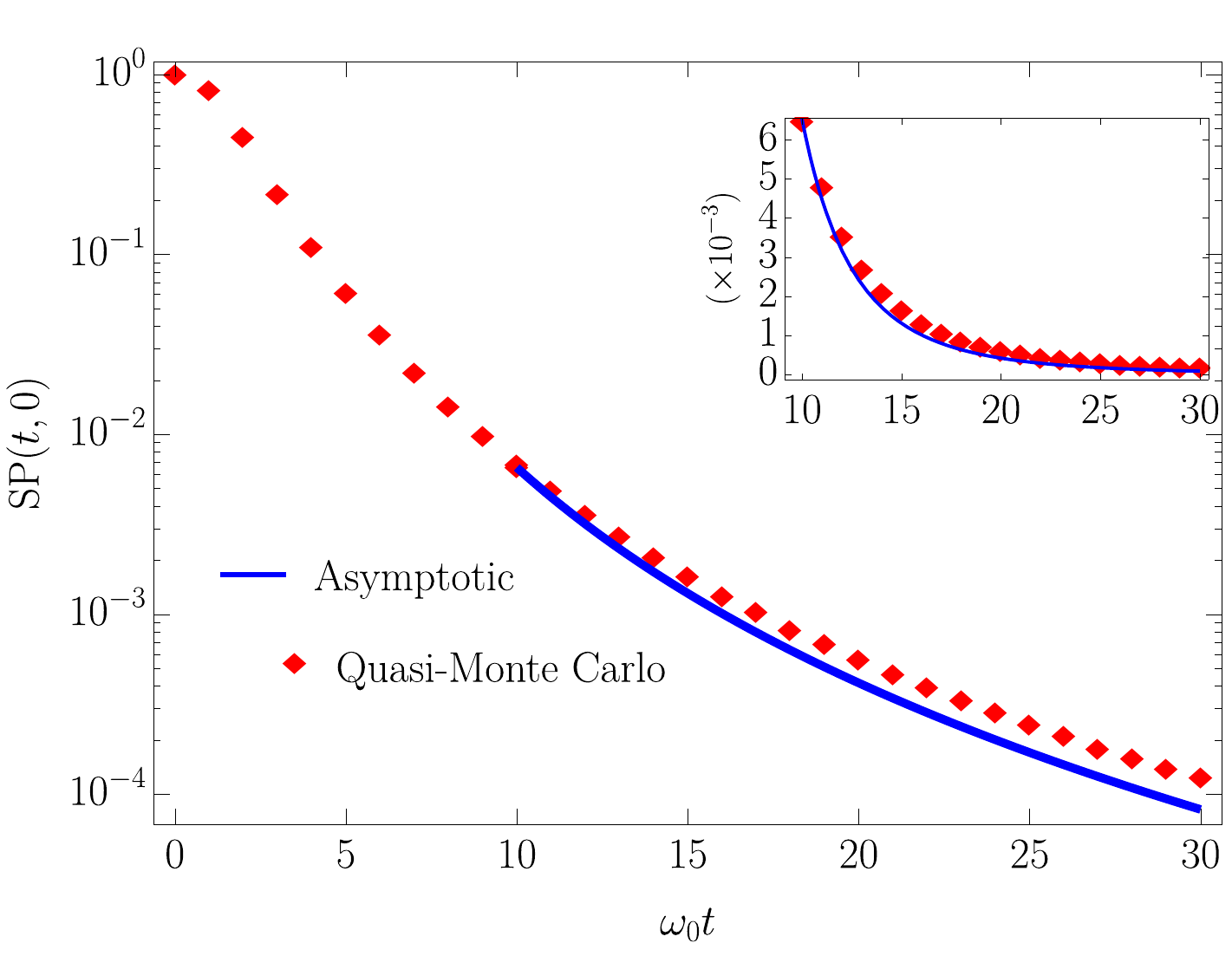}
\par\end{centering}
\caption{\label{fig:SP-LL-TrapQuench}The numerical calculation versus the
analytic prediction of the survival probability after a quench of
the harmonic trap in the long-range LL model. The values of parameters:
$\hbar=1$, $N=4$, $\omega_{0}=1$ and $c_{0}=-1$. The long-time
decay is characterized by asymptotic expression in Eq.~(\ref{eq:SP-LL}).}
\end{figure}



\section{Discussion}

We have shown how finding the parent Hamiltonian of a complex-valued
TDJA leads to the systematic generation of exactly-solvable one-dimensional
strongly-correlated systems away from equilibrium. While we have focused
on spinless bosonic systems, in which the many-body wavefunction is
permutation-symmetric under particle exchange, it is natural to consider
generalizations to fermionic and anyonic systems. In the presence
of hard-core interactions, a natural avenue to do so is via the Bose-Fermi
mapping \citep{Girardeau60} and its generalization to anyonic systems
\citep{Girardeau06}. Such dualities establish a map between bosonic
states and states with other quantum exchange statistics, both in
and out of equilibrium \citep{Girardeau00,minguzzi2005exactcoherent,delcampo08}.
For finite interactions with coupling strength $c$, one can make
use of the generalizations of the Bose-Fermi mapping that hold at
strong coupling, following Gaudin \citep{Gaudin14}. This approach
has already been fruitful in the study of the time-dependent LL gas
using a $1/c$ expansion \citep{Buljan08,Jukic08,Pezer09} and paves
the way to find time-dependent fermionic and anyonic models dual to
the long-range Lieb-Liniger model we have discussed. Yet a different
approach is to rely on the exact mapping uncovered by Bethe-ansatz
between models with different particle statistics \citep{Kundu09,Batchelor06,Batchelor06b,Batchelor08}.

Likewise, further studies can be envisioned by generalizing our results
to higher spatial dimensions \citep{CalogeroMarchioro75,KhareRay97,BeauDC21},
mixtures \citep{GirardeauMinguzzi07,Harshman17}, truncated interactions
\citep{JainKhare99,Pittman17}, and particles with spin or internal
structure \citep{Vacek94,Kawakami94,Girardeau04}.

One may further wonder whether the exact evolutions discussed here
admit an integrable structure analogous to that of time-dependent
quantum Hamiltonians constructed from a zero-curvature representation
in the space of parameters \citep{Sinitsyn18}.

\section{Conclusion }

In summary, we have introduced a family of one-dimensional time-dependent
many-body Hamiltonians whose nonequilibrium strongly-correlated dynamics
is exactly described by the complex-valued TDJA. To guarantee the
Hermicity of the parent Hamiltonian, we derive consistency conditions
for the one-body and two-body functions determining the Jastrow ansatz.
In doing so, the description of the dynamics is remarkably reduced
and boils down to determining the flow of the coupling constants in
these models in terms of simple ordinary differential equations.

We illustrate our findings in four classes of examples, including
the celebrated Calogero-Sutherland, hyperbolic, and Lieb-Liniger models
in the presence of a driven harmonic trap. We show that these results
can be applied to quantum control, and specifically, to the engineering
shortcuts to adiabaticity, leading to protocols  \textcolor{blue}{that steer the adiabatic
evolution of strongly correlated systems described by the real-valued TIJA. The main findings are summarized
in Table~\ref{tab:Summary-of-finding}. For scale-invariant
dynamics with $\eta(t)=0$, we found that one can relate the late-time density  and  momentum distributions to the
initial density distribution, as  was previously established
for Tonks-Girardeau gases~\citep{minguzzi2005exactcoherent,Dupays22}.} Finally, we show that
our results allow us to describe the exact nonequilibrium dynamics
following a quench of the interactions, as discussed in the generalized
Calogero-Sutherland models with logarithmic interactions and the long-range
Lieb-Liniger model describing bosons subject to contact and Coulomb
interactions.

Our findings provide valuable analytic solutions for the dynamics
of one-dimensional strongly correlated quantum systems of continuous
variables, making it possible to study in these systems the quantum
work statistics, Loschmidt echos \citep{delcampo16,LvZhang20}, orthogonality
catastrophe \citep{Fogarty20,Mackel22}, fidelity susceptibility \citep{Gu10},
quantum speed limits \citep{XuLi20,delcampo21} and other bounds on
quantum evolution \citep{GongHamazaki22}. We hope our results can
inspire further experimental studies of non-equilibrium dynamics with
strongly correlated quantum matter and be used as a benchmark for
numerical methods, quantum simulators, and quantum computers.


\section{Acknowledgement}

We are grateful to Yun-Feng Xiong and Hongzhe Zhou for helpful discussions
on Monte Carlo simulations and L\'eonce Dupays for comments on this
work.

\vspace{1mm}
\onecolumngrid

\appendix

\section{\label{sec:Hermicitiy}Details on the consistency conditions}

\subsection{The three-body potential $v_{3}^{(ijk)}$}

We introduce the prepotential~\citep{Polychronakos92} 
\begin{equation}
w(x,\,t)=f^{\prime}(x,\,t)/f(x,\,t)=\Gamma^{\prime}(x,\,t)+\text{i}\theta^{\prime}(x,\,t).
\end{equation}
Note that $w$, $\Gamma^{\prime}$ and $\theta^{\prime}$ are all
odd functions of $x$. However, unlike the case of the TIJA wave function,
the ansatz is a complex-valued function. The condition for the complex
three-body potential term in Eq.~(\ref{eq:scrH-def}) to reduce to a
two-body complex potential is the same as in the case of the TIJA
wave function~\citep{Calogero75,Sutherland04,yang2022one}, i.e.,
$w$ must satisfy 
\begin{equation}
w(x,\,t)w(y,\,t)+w(y,\,t)w(z,\,t)+w(z,t)w(x,t)=h(x,t)+h(y,t)+h(z,t),\label{eq:w-property}
\end{equation}
where $h(x,\,t)$ is some even function in $x$. This leads to the
following differential equation~\citep{Calogero75} 
\begin{equation}
c_{1}(t)w^{\prime\prime\prime}(x,t)-12c_{2}(t)w^{\prime}(x,t)+6[w(x,t)]^{2}=c_{0}(t),
\end{equation}
where $c_{1}(t)$,$c_{2}(t)$, and $c(t)$ are functions of $t$ exclusively.
Upon making the change of variables $w^{\prime}(x,t)=-c_{1}(t)u(x,t)+c_{2}(t)$,
we find 
\begin{equation}
u^{\prime\prime}(x,t)=6[u(x,t)]^{2}-[6c_{2}^{2}(t)+c_{0}(t)]/c_{1}^{2}(t).
\end{equation}
In the case where $w(x,\,t)=\int^{x}u(y,t)dy=\Gamma^{\prime}(x,\,t)$
is real, i.e., $\theta^{\prime}(x,\,t)=0$, its general solution can
be expressed in terms of the Weierstrass zeta function. When 
\begin{equation}
\Gamma^{\prime}(x,\,t)=\frac{\lambda(t)}{x},\quad\lambda(t)a(t)\coth[a(t)x],\quad c(t)\text{sgn}(x),
\end{equation}
which corresponds to 
\begin{equation}
e^{\Gamma(x,\,t)}\propto|x|^{\lambda(t)},\quad|\sinh[a(t)x]|^{\lambda(t)},\quad\exp[c(t)|x|].
\end{equation}
the three-body potential is a constant. Note that replacing $x\to b(t)x$
in $e^{\Gamma(x,\,t)}$ will not changes the form of the $\Gamma^{\prime}(x,\,t)$.

One can show that for $e^{\Gamma(x,\,t)}\propto|x|^{\lambda(t)}$
\begin{equation}
W_{3}(t)=0.
\end{equation}
For $e^{\Gamma(x,\,t)}\propto|\sinh[c(t)x]|^{\lambda(t)}$, 
\begin{equation}
W_{3}(t)=\frac{N(N-1)(N-2)\lambda^{2}(t)c^{2}(t)}{6}.
\end{equation}
In addition, for $e^{\Gamma(x,\,t)}\propto\exp[c(t)|x|]$ 
\begin{equation}
W_{3}(t)=\frac{N(N-1)(N-2)c^{2}(t)}{6}.
\end{equation}
In general, it may be complicated to find complex-valued functions
$w$ that satisfy Eq.~(\ref{eq:w-property}) in terms of elementary
functions. However, we shall discuss a simple but non-trivial case
where the real part of $w(x,\,t)$ is proportional to its imaginary
part. This results in 
\begin{equation}
\theta^{\prime}(x,\,t)=\eta(t)\Gamma^{\prime}(x,\,t).
\end{equation}
Upon integration of both sides, we obtain 
\begin{align}
\theta(x,\,t) & =\eta(t)\Gamma(x,\,t),\label{eq:theta-c}\\
w(x,\,t) & =\Gamma^{\prime}(x,\,t)[1+\text{i}\eta(t)].\label{eq:wc}
\end{align}
When $e^{\Gamma(x,\,t)}$ takes the functional form of the ground
state wave functions of the CS, Hyperbolic, and LL models, the three-body
potential is then given by Eq.~(\ref{eq:U3-special}) in the main
text. As discussed there, we have effectively absorbed a position-independent
phase in $\theta(x,\,t)$ into the definition of $\phi(x,\,t)$.

\subsection{Two-body potential $v_{2}^{(ij)}$}

It is straightforward to compute 
\begin{align}
f_{ij}^{\prime} & =\left(\Gamma_{ij}^{\prime}+\text{i}\theta_{ij}^{\prime}\right)f_{ij}=(1+\text{i}\eta)\Gamma_{ij}^{\prime}f_{ij},\\
f_{ij}^{\prime\prime} & =(1+\text{i}\eta)\Gamma_{ij}^{\prime\prime}f_{ij}+(1+\text{i}\eta)^{2}\Gamma_{ij}^{\prime2}f_{ij},\\
\dot{f}_{ij} & =(\dot{\Gamma}_{ij}+\text{i}\dot{\theta}_{ij})f_{ij},\\
g_{i}^{\prime} & =(\Upsilon_{i}^{\prime}+\text{i}\phi_{i}^{\prime})g_{i}.
\end{align}
Thus, we find 
\begin{align}
\text{Re}v_{2}^{(ij)} & =\Gamma_{ij}^{\prime\prime}+(1-\eta^{2})\Gamma_{ij}^{\prime2}+\Gamma_{ij}^{\prime}[(\Upsilon_{i}^{\prime}-\Upsilon_{j}^{\prime})-\eta(\phi_{i}^{\prime}-\phi_{j}^{\prime})]-\frac{m}{\hbar}\dot{\theta}_{ij},\label{eq:ReU2}\\
\text{Im}v_{2}^{(ij)} & =\eta\Gamma_{ij}^{\prime\prime}+2\eta\Gamma_{ij}^{\prime2}+\Gamma_{ij}^{\prime}[\eta(\Upsilon_{i}^{\prime}-\Upsilon_{j}^{\prime})+(\phi_{i}^{\prime}-\phi_{j}^{\prime})]+\frac{m}{\hbar}\dot{\Gamma}_{ij}.\label{eq:ImU2}
\end{align}
Note that according to Eq.~(\ref{eq:theta-c}), it follows that 
\begin{equation}
\dot{\theta}_{ij}=\dot{\eta}\Gamma_{ij}+\eta\dot{\Gamma}_{ij}.
\end{equation}
As we have seen in the previous subsection, the three-body potential
term is at most a complex-valued constant. The two-body terms, if
depending on $x_{ij}$, in any case, cannot be canceled by the one-body
potential. Then the Hermiticity of the Hamiltonian requires that $\text{Im}U_{2}^{(ij)}$
is a function of time only and independent of $x_{ij}$. Since $\Gamma_{ij}^{\prime}(x)$
is an odd function, we note that $\Gamma_{ij}^{\prime\prime}$, $\Gamma_{ij}^{\prime2}$
and $\dot{\Gamma}_{ij}$ are even. This implies that 
\begin{equation}
\eta(\Upsilon_{i}^{\prime}-\Upsilon_{j}^{\prime})+(\phi_{i}^{\prime}-\phi_{j}^{\prime})=\frac{1}{\Gamma_{ij}^{\prime}}\left(-\eta\Gamma_{ij}^{\prime\prime}-2\eta\Gamma_{ij}^{\prime2}-\frac{m}{\hbar}\dot{\Gamma}_{ij}+\dot{\widetilde{\mathcal{N}}}_{2}\right)\label{eq:2-body-im}
\end{equation}
must be an even function of $x_{ij}$ only and independent of $\bar{x}_{ij}=(x_{i}+x_{j})/2$
and that $\dot{\mathcal{N}}_{2}(t)$ is a function of time only. 
\begin{lem}
For a differentiable function $F(x)$, the only possibility for $F(x)-F(y)$
to be only dependent on $x-y$ is that $F(x)$ is a linear function
of $x$. 
\end{lem}
\begin{proof}
The proof is straightforward: We denote $F(x)-F(y)=G(x-y)$. Performing
a Taylor expansion, we find 
\begin{equation}
F^{\prime}(x)\epsilon=G^{\prime}(0)\epsilon,
\end{equation}
which dictates the $F^{\prime}(x)$ must be a constant. 
\end{proof}
According to the above Lemma, we conclude that 
\begin{equation}
\eta(t)\Lambda_{k}^{\prime}(t)+\phi_{k}^{\prime}(t)=-\left[\frac{m}{\hbar}\mathscr{C}_{2}(t)x_{k}+\frac{m^{2}}{\hbar^{2}}\mathscr{D}_{2}(t)\right],\label{eq:linear-func}
\end{equation}
where $\mathscr{C}_{2}(t)$ and $\mathscr{D}_{2}(t)$ are real-valued
functions of time. Integrating both sides of Eq.~(\ref{eq:linear-func})
over $x$, yields

\begin{equation}
\phi_{k}(t)=-\left[\eta(t)\Lambda_{k}(t)+\frac{m}{2\hbar}\mathscr{C}_{2}(t)x_{k}^{2}+\frac{m^{2}}{\hbar^{2}}\mathscr{D}_{2}(t)x_{k}\right]+\tau(t),\label{eq:1-body-PA}
\end{equation}
where $\tau(t)$ is the overall phase angle.

Substituting Eq.~(\ref{eq:linear-func}) into Eq.~(\ref{eq:2-body-im})
yields 
\begin{equation}
\eta(t)\Gamma_{ij}^{\prime\prime}+2\eta(t)\Gamma_{ij}^{\prime2}+\frac{m}{\hbar}\dot{\Gamma}_{ij}-\frac{m}{\hbar}\mathscr{C}_{2}(t)\Gamma_{ij}^{\prime}x_{ij}=\dot{\widetilde{\mathcal{N}}}_{2}(t).\label{eq:Im-2-body}
\end{equation}
With Eq.~(\ref{eq:linear-func}), Eq.~(\ref{eq:ReU2}) becomes 
\begin{eqnarray}
\text{Re}v_{2}^{(ij)} & =& \Gamma_{ij}^{\prime\prime}+[1-\eta^{2}(t)]\Gamma_{ij}^{\prime2}+[1+\eta^{2}(t)]\Gamma_{ij}^{\prime}(\Lambda_{i}^{\prime}-\Lambda_{j}^{\prime})\nonumber \\
 & & +\frac{m}{\hbar}\eta(t)\mathscr{C}_{2}(t)\Gamma_{ij}^{\prime}x_{ij}-\frac{m}{\hbar}\frac{d}{dt}[\eta(t)\Gamma_{ij}].\label{eq:ReU2-compact}
\end{eqnarray}

\subsection{One-body potential}

It is straightforward to compute 
\begin{align}
g_{i}^{\prime\prime} & =(\Lambda_{i}^{\prime\prime}+\text{i}\phi_{i}^{\prime\prime})g_{i}+(\Lambda_{i}+\text{i}\phi_{i}^{\prime})^{2}g_{i},\\
\dot{g}_{i} & =(\dot{\Lambda}_{i}+\text{i}\dot{\phi}_{i})g_{i}.
\end{align}
Then, according to $v_{1}^{(i)}=g_{i}^{\prime\prime}/g_{i}+2\text{i}m\dot{g}_{i}/(\hbar g_{i})$,
we find 
\begin{align}
\text{Re}v_{1}^{(i)} & =\Lambda_{i}^{\prime\prime}+\Lambda_{i}^{\prime2}-\phi_{i}^{\prime2}-\frac{2m}{\hbar}\dot{\phi}_{i},\label{eq:ReU1}\\
\text{Im}v_{1}^{(i)} & =\phi_{i}^{\prime\prime}+2\Lambda_{i}^{\prime}\phi_{i}^{\prime}+\frac{2m}{\hbar}\dot{\Lambda}_{i}.\label{eq:ImU1}
\end{align}
The Hermiticity of the one-body potential dictates that $\text{Im}U_{1}^{(i)}$
is only a function of time and independent of $x_{i}$, i.e, 
\begin{equation}
\phi_{i}^{\prime\prime}+2\Lambda_{i}^{\prime}\phi_{i}^{\prime}+\frac{2m}{\hbar}\dot{\Lambda}_{i}=\dot{\widetilde{\mathcal{N}}}_{1}.\label{eq:ImU1-diff}
\end{equation}
Substituting Eq.~(\ref{eq:linear-func}) into Eq.~(\ref{eq:ImU1-diff})
yields, 
\begin{align}
-\eta(t)\Lambda_{i}^{\prime\prime}-2\eta(t)\Lambda_{i}^{\prime2}-\frac{2m}{\hbar}\mathscr{C}_{2}(t)\Lambda_{i}^{\prime}x_{i}-\frac{2m^{2}}{\hbar^{2}}\mathscr{D}_{2}(t)\Lambda_{i}^{\prime}+\frac{2m}{\hbar}\dot{\Lambda}_{i} & =\dot{\widetilde{\mathcal{N}}}_{1}(t)+\frac{m}{\hbar}\mathscr{C}_{2}(t).\label{eq:Im-1-body}
\end{align}
Similarly, substitution of Eq.~(\ref{eq:linear-func}) into Eq.~(\ref{eq:ReU1})
yields 
\begin{eqnarray}
\text{Re}v_{1}^{(i)} & =&\Lambda_{i}^{\prime\prime}+\Lambda_{i}^{\prime2}[1-\eta^{2}(t)]-\frac{m^{2}}{\hbar^{2}}\mathscr{C}_{2}^{2}(t)x_{i}^{2}-\frac{m^{4}}{\hbar^{4}}\mathscr{D}_{2}^{2}(t)\nonumber \\
 && -\frac{2m^{3}}{\hbar^{3}}\mathscr{C}_{2}(t)\mathscr{D}_{2}(t)x_{i}-2\eta\Lambda_{i}^{\prime}\left[\frac{m}{\hbar}\mathscr{C}_{2}(t)x_{i}+\frac{m^{2}}{\hbar^{2}}\mathscr{D}_{2}(t)\right]\nonumber \\
 & &+\frac{2m}{\hbar}\left[\frac{d}{dt}[\eta(t)\Lambda_{i}]+\frac{m}{2\hbar}\dot{\mathscr{C}_{2}}(t)x_{i}^{2}+\frac{m^{2}}{\hbar^{2}}\dot{\mathscr{D}}_{2}(t)x_{i}-\dot{\tau}(t)\right].\label{eq:ReU1-compact}
\end{eqnarray}

\subsection{The general consistency conditions}

To summarize, Eqs.~(\ref{eq:Im-2-body},~\ref{eq:Im-1-body}) are
the generic consistency conditions for the parent Hamiltonian of time-dependent
Jastrow ansatz when the trial wave function $e^{\Gamma_{ij}(t)}$
is restricted to the CS, hyperbolic, and LL models. Eqs.~(\ref{eq:Im-2-body},~\ref{eq:Im-1-body})
can be further recast into

\begin{equation}
\eta(t)\Gamma_{ij}^{\prime\prime}+2\eta(t)\Gamma_{ij}^{\prime2}+\frac{m}{\hbar}\dot{\Gamma}_{ij}-\frac{m}{\hbar}\mathscr{C}_{2}(t)\Gamma_{ij}^{\prime}x_{ij}=\dot{\widetilde{\mathcal{N}}}_{2}(t),\label{eq:2-body-consist}
\end{equation}
\begin{equation}
\eta(t)\Lambda_{i}^{\prime\prime}+2\eta(t)\Lambda_{i}^{\prime2}+\frac{2m\mathscr{C}_{2}(t)}{\hbar}\left(\Lambda_{i}^{\prime}x_{i}+\frac{1}{2}\right)+\frac{2m^{2}\mathscr{D}_{2}(t)}{\hbar^{2}}\Lambda_{i}^{\prime}-\frac{2m}{\hbar}\dot{\Lambda}_{i}=-\dot{\widetilde{\mathcal{N}}}_{1}(t).\label{eq:1-body-consist}
\end{equation}

\section{\label{sec:Consistency-condition-Harmonic}Consistency conditions
for harmonic trap}

For systems confined in a harmonic trap, we would like Eq.~(\ref{eq:ReU1-compact})
to be at most quadratic in $x_{k}$. Thus, we take 
\begin{equation}
\Lambda_{k}=-\frac{m\omega(t)x_{k}^{2}}{2\hbar}.
\end{equation}
In this case, the one-body consistency condition Eq.~(\ref{eq:1-body-consist})
simplifies dramatically, giving $\mathscr{D}_{2}(t)=0$ and 
\begin{align}
\mathscr{C}_{2}(t) & =\eta\omega(t)+\frac{\dot{\omega}(t)}{2\omega(t)},\label{eq:1-body-consist-harmonic}\\
\dot{\mathcal{N}}_{1}(t) & =-\frac{m}{\hbar}[\mathscr{C}_{2}(t)-\eta(t)\omega(t)]=-\frac{m\dot{\omega}(t)}{2\hbar\omega(t)}.
\end{align}
The one-body phase angle~(\ref{eq:1-body-PA}) now becomes 
\begin{equation}
\phi_{k}(t)=-\frac{m}{2\hbar}\left[\mathscr{C}_{2}(t)-\eta(t)\omega(t)\right]x_{k}^{2}+\tau(t)=-\frac{m\dot{\omega}(t)}{4\hbar\omega(t)}x_{k}^{2}+\tau(t).\label{eq:phi-1PA-harmonic}
\end{equation}
Therefore, substituting Eq.~(\ref{eq:phi-1PA-harmonic}) into Eq.~(\ref{eq:ReU1}),
the real part of the one-body potential becomes 
\begin{equation}
\text{Re}v_{1}^{(i)}=\frac{m^{2}\Omega^{2}(t)}{\hbar^{2}}x_{i}^{2}-\frac{m}{\hbar}\vartheta(t),
\end{equation}
where $\Omega(t)$ and $\vartheta(t)$ are defined in Eq.~(\ref{eq:omg-def})
and Eq.~(\ref{eq:vartheta-def}) in the main text, respectively.
One can perform a sanity check by substituting Eq.~(\ref{eq:1-body-consist-harmonic})
into Eq.~(\ref{eq:ReU1-compact}) and find that the real part of
the one-body potential is indeed independent of $\eta(t)$ and Eq.~(\ref{eq:omg-def})
is reproduced.

Eq.~(\ref{eq:ReU2-compact}) becomes 
\begin{eqnarray}
\text{Re}v_{2}^{(ij)} & =&\Gamma_{ij}^{\prime\prime}+[1-\eta^{2}(t)]\Gamma_{ij}^{\prime2}-\frac{m\omega(t)}{\hbar}[1+\eta^{2}(t)]\Gamma_{ij}^{\prime}x_{ij}\nonumber \\
 & &+\frac{m}{\hbar}\eta(t)\left[\eta\omega(t)+\frac{\dot{\omega}(t)}{2\omega(t)}\right]\Gamma_{ij}^{\prime}x_{ij}-\frac{m}{\hbar}\frac{d}{dt}[\eta(t)\Gamma_{ij}]\nonumber \\
 & =&\Gamma_{ij}^{\prime\prime}+[1-\eta^{2}(t)]\Gamma_{ij}^{\prime2}+\frac{m\omega(t)}{\hbar}\left[\frac{\eta(t)\dot{\omega}(t)}{2\omega^{2}(t)}-1\right]\Gamma_{ij}^{\prime}x_{ij}-\frac{m}{\hbar}\frac{d}{dt}[\eta(t)\Gamma_{ij}]\nonumber \\
 & =&\Gamma_{ij}^{\prime\prime}+[1-\eta^{2}(t)]\Gamma_{ij}^{\prime2}-\frac{m\varpi(t)}{\hbar}\Gamma_{ij}^{\prime}x_{ij}-\frac{m}{\hbar}\frac{d}{dt}[\eta(t)\Gamma_{ij}],
\end{eqnarray}
where $\varpi(t)$ is defined in Eq.~(\ref{eq:omgbar-def}) in the
main text.

\section{\label{sec:PHJ-rotating}The parent Hamiltonian of the real-valued
TDJA}

One can rewrite $\mathscr{U}_{P}(t)$ as 
\begin{align}
\mathscr{U}_{P}(\bm{x},\,t) & =\exp\left[\text{i}\left(\eta(t)\sum_{i<j}\Gamma_{ij}-\sum_{k}\frac{m\dot{\omega}(t)}{4\hbar\omega(t)}x_{k}^{2}+N\tau(t)\right)\right]\label{eq:scrU-rep1}\\
 & =\exp\left[\text{i}\left(\eta(t)\sum_{l\neq i}\Gamma_{il}-\frac{m\dot{\omega}(t)}{4\hbar\omega(t)}x_{i}^{2}\right)+\text{i}\left(\eta(t)\sum_{k<l,\,k,l\neq i}\Gamma_{kl}-\frac{m\dot{\omega}(t)}{4\hbar\omega(t)}\sum_{l\neq i}x_{l}^{2}+\tau(t)\right)\right].\label{eq:scrU-rep2}
\end{align}
With Eqs.~(\ref{eq:scrU-rep1},~\ref{eq:scrU-rep2}), one finds 
\begin{equation}
\dot{\mathscr{U}_{P}}(\bm{x},t)=\text{i}\mathscr{U}_{P}(\bm{x},t)\left(\sum_{i<j}\frac{d}{dt}[\eta(t)\Gamma_{ij}]-\frac{1}{2\hbar}\frac{d}{dt}\left[\frac{m\dot{\omega}(t)}{2\omega(t)}\right]\sum_{k}x_{k}^{2}+N\dot{\tau}(t)\right),\label{eq:scrU-dot}
\end{equation}
\begin{align}
\mathscr{U}_{P}^{\dagger}(\bm{x},t)p_{i}\mathscr{U}_{P}(\bm{x},t) & =p_{i}-\text{i}[\eta(t)\sum_{l,\,l\neq k}\Gamma_{il}-\frac{m\dot{\omega}(t)}{4\hbar\omega(t)}x_{i}^{2},\,p_{i}]\nonumber \\
 & =p_{i}+\hbar\eta(t)\sum_{l\neq i}\Gamma_{il}^{\prime}-\frac{m\dot{\omega}(t)}{2\omega(t)}x_{i}.\label{eq:p-prime}
\end{align}
Substituting Eqs.~(\ref{eq:scrU-dot},~\ref{eq:p-prime}) into Eq.~(\ref{eq:U-connection})
in the main text, this leads to 
\begin{eqnarray}
\mathscr{H}^{\prime}(t) & =&\frac{1}{2m}\sum_{i}p_{i}^{2}+\frac{\hbar^{2}}{2m}\eta^{2}(t)\sum_{i}\left[\sum_{j\neq i}\Gamma_{ij}^{\prime}\right]\left[\sum_{k\neq i}\Gamma_{ik}^{\prime}\right]\nonumber \\
 & &+\frac{1}{2m}\sum_{i}\bigg\{\hbar\eta(t)\sum_{j\neq i}\Gamma_{ij}^{\prime}-\frac{m\dot{\omega}(t)}{2\omega(t)}x_{i},\,p_{i}\bigg\}-\frac{\hbar\dot{\omega}(t)\eta(t)}{2\omega(t)}\sum_{i}\sum_{j\neq i}\Gamma_{ij}^{\prime}x_{i}\nonumber \\
 & &+\frac{1}{2}m\omega^{2}(t)\sum_{k}x_{k}^{2}+\frac{\hbar^{2}}{m}\sum_{i<j}\left(\Gamma_{ij}^{\prime\prime}+[1-\eta^{2}(t)]\Gamma_{ij}^{\prime2}\right)\nonumber \\
 & &-\hbar\varpi(t)\sum_{i<j}\Gamma_{ij}^{\prime}x_{ij}+\tilde{\mathscr{E}}^{\prime}(t),\label{eq:H-prime-exp}
\end{eqnarray}
where 
\begin{align}
\tilde{\mathscr{E}}^{\prime}(t) & \equiv\mathscr{E}(t)+N\hbar\dot{\tau}(t)=-\frac{1}{2}N\hbar\omega(t)+\frac{\hbar^{2}}{m}[1-\eta^{2}(t)]W_{3}(t).
\end{align}
Remarkably, upon transforming $\mathscr{H}(t)\to\mathscr{H}^{\prime}(t)$,
the frequency of the trap transforms as $\Omega(t)\to\omega(t)$~\citep{delcampo2013shortcuts}.
Furthermore, the term $\sum_{i<j}\frac{d}{dt}[\eta(t)\Gamma_{ij}]$
in $\mathscr{H}(t)$ {[}see Eq.~(\ref{eq:TDPHJ}){]} cancels. Finally,
we note that 
\begin{equation}
\sum_{i}\left[\sum_{j\neq i}\Gamma_{ij}^{\prime}\right]\left[\sum_{k\neq i}\Gamma_{ik}^{\prime}\right]=\sum_{i\neq j\ne k\neq i}\Gamma_{ij}^{\prime}\Gamma_{ik}^{\prime}+\sum_{i\neq j}\Gamma_{ij}^{\prime2}.
\end{equation}
Furthermore, we find 
\begin{equation}
\sum_{\substack{i\neq j\ne k\neq i}
}\Gamma_{ij}^{\prime}\Gamma_{ik}^{\prime}=-\sum_{\substack{i\neq j\ne k\neq i}
}\Gamma_{ij}^{\prime}\Gamma_{jk}^{\prime},
\end{equation}
where we have changed the dummy indices $i$ and $j$. Upon changing
the dummy indices $j$ and $k$, we find 
\begin{equation}
\sum_{\substack{i\neq j\ne k\neq i}
}\Gamma_{ij}^{\prime}\Gamma_{jk}^{\prime}=\sum_{\substack{i\neq j\ne k\neq i}
}\Gamma_{ik}^{\prime}\Gamma_{kj}^{\prime}=\sum_{\substack{i\neq j\ne k\neq i}
}\Gamma_{ki}^{\prime}\Gamma_{jk}^{\prime}.
\end{equation}
Therefore, it follows that 
\begin{align}
\sum_{\substack{i\neq j\ne k\neq i}
}\Gamma_{ij}^{\prime}\Gamma_{ik}^{\prime} & =-\frac{1}{3}\sum_{\substack{i\neq j\ne k\neq i}
}\left(\Gamma_{ij}^{\prime}\Gamma_{jk}^{\prime}+\Gamma_{ij}^{\prime}\Gamma_{ki}^{\prime}+\Gamma_{ki}^{\prime}\Gamma_{jk}^{\prime}\right)\nonumber \\
 & =-2\sum_{\substack{i<j<k}
}\left(\Gamma_{ij}^{\prime}\Gamma_{jk}^{\prime}+\Gamma_{ij}^{\prime}\Gamma_{ki}^{\prime}+\Gamma_{ki}^{\prime}\Gamma_{jk}^{\prime}\right).
\end{align}
According to Eq.~(\ref{eq:v3C-def}), 
\begin{align}
\sum_{i}\left[\sum_{j\neq i}\Gamma_{ij}^{\prime}\right]\left[\sum_{k\neq i}\Gamma_{ik}^{\prime}\right] & =-2\sum_{\substack{i<j<k}
}\left(\Gamma_{ij}^{\prime}\Gamma_{jk}^{\prime}+\Gamma_{ij}^{\prime}\Gamma_{ki}^{\prime}+\Gamma_{ki}^{\prime}\Gamma_{jk}^{\prime}\right)+2\sum_{i<j}\Gamma_{ij}^{\prime2}\nonumber \\
 & =2W_{3}(t)+2\sum_{i<j}\Gamma_{ij}^{\prime2}.\label{eq:3-body-term}
\end{align}
Furthermore, we note that since $\Gamma_{ij}^{\prime}$ is skew-symmetric
in the indices $i$ and $j$, 
\begin{equation}
\sum_{i}\sum_{j\neq i}\Gamma_{ij}^{\prime}(t)x_{i}\equiv\frac{1}{2}\sum_{i}\sum_{j\neq i}\Gamma_{ij}^{\prime}(t)x_{ij}=\sum_{i<j}\Gamma_{ij}^{\prime}(t)x_{ij}.\label{eq:2-body}
\end{equation}
Similarly, 
\begin{align}
\sum_{i}\bigg\{\sum_{j\neq i}\Gamma_{ij}^{\prime}(t),\,p_{i}\bigg\} & =\sum_{i\neq j}\Gamma_{ij}^{\prime}(t)p_{i}+\sum_{i\neq j}p_{i}\Gamma_{ij}^{\prime}(t)\nonumber \\
 & =\frac{1}{2}\sum_{i\neq j}\Gamma_{ij}^{\prime}(t)(p_{i}-p_{j})+\frac{1}{2}\sum_{i\neq j}(p_{i}-p_{j})\Gamma_{ij}^{\prime}(t)\nonumber \\
 & =\frac{1}{2}\sum_{i\neq j}\{\Gamma_{ij}^{\prime}(t),\,p_{ij}\}\nonumber \\
 & =\sum_{i<j}\{\Gamma_{ij}^{\prime}(t),\,p_{ij}\}.\label{eq:Gamma-P-term}
\end{align}
Substituting Eqs.~(\ref{eq:3-body-term},~\ref{eq:2-body},~\ref{eq:Gamma-P-term})
into Eq.~(\ref{eq:H-prime-exp}) yields 
\begin{eqnarray}
\mathscr{H}^{\prime}(t) & =&\frac{1}{2m}\sum_{i}p_{i}^{2}+\frac{1}{2}m\omega^{2}(t)\sum_{k}x_{k}^{2}+\frac{\hbar^{2}}{m}\sum_{i<j}\left(\Gamma_{ij}^{\prime\prime}+\Gamma_{ij}^{\prime2}\right)\nonumber \\
 & &+\frac{\hbar\eta(t)}{2m}\sum_{i<j}\{\Gamma_{ij}^{\prime}(t),\,p_{ij}\}-\frac{\dot{\omega}(t)}{4\omega(t)}\sum_{i}\{x_{i},\,p_{i}\}-\hbar\left[\varpi(t)+\frac{\dot{\omega}(t)\eta(t)}{2\omega(t)}\right]\sum_{i<j}\Gamma_{ij}^{\prime}x_{ij}+\mathscr{E}_{0}^{\prime}(t).
\end{eqnarray}


\twocolumngrid

\bibliographystyle{apsrev4-1}
\bibliography{BAIQS_lib}

\end{document}